\documentclass[11pt]{article}
\usepackage{amsmath,amssymb,amsthm,amsfonts,latexsym,xspace,enumerate,graphicx,color,eucal,upgreek,mathtools}
\usepackage[backref,colorlinks,bookmarks=true]{hyperref}

\newtheorem{theorem}{Theorem}[section]
\newtheorem*{namedtheorem}{\theoremname}
\newcommand{\theoremname}{testing}

\newtheorem{lemma}[theorem]{Lemma}

\newtheorem{claim}[theorem]{Claim}

\newtheorem{observation}{Observation}
\newtheorem{corollary}[theorem]{Corollary}

\newtheorem{question}[theorem]{Question}
\newtheorem{cons}{Constraints}

\newtheorem{vars}{Variables}
\theoremstyle{definition}
\newtheorem{definition}[theorem]{Definition}

\theoremstyle{plain}
\newtheorem{Alg}{Algorithm}

\renewenvironment{proof}{\noindent{\textbf{Proof:}}} {$\blacksquare$\vskip \belowdisplayskip}

\newcommand{\ignore}[1]{}

\newcommand{\poly}{\mathrm{poly}}

\newcommand{\calC}{{\mathcal{C}}}

\newcommand{\calH}{\mathcal{H}}

\newcommand{\calR}{\mathcal{R}}

\newcommand{\norm}[1]{\left\lVert #1 \right\rVert}

\newcommand{\set}[1]{\left\{ #1 \right\}}

\title{Computing a Nonnegative Matrix Factorization \--- Provably}
\author{Sanjeev Arora \thanks{Princeton University, Computer Science Department and Center for Computational Intractability.
Email: {\tt arora@cs.princeton.edu}. This work is supported by the NSF grants CCF-0832797 and CCF-1117309.} \and
Rong Ge \thanks{Princeton University, Computer Science Department and Center for Computational Intractability.
Email: {\tt rongge@cs.princeton.edu}. This work is supported by the NSF grants CCF-0832797 and CCF-1117309.} \and
Ravi Kannan \thanks{Microsoft Research labs., India. Email:{\tt kannan@microsoft.com}} \and
Ankur Moitra \thanks{Institute for Advanced Study, School of Mathematics.
Email: {\tt moitra@ias.edu}.
Research supported in part 
by NSF grant No.
DMS-0835373 and by an NSF Computing and Innovation Fellowship.}
}

\begin{document}
\maketitle

\begin{abstract}
In the {\em Nonnegative Matrix Factorization} (NMF) problem we are given an $n \times m$ nonnegative matrix $M$ and an integer $r > 0$. Our goal is to express $M$ as $A W$ where $A$ and $W$ are nonnegative matrices of size $n \times r$ and $r \times m$ respectively. In some applications, it makes sense to ask instead for the product $AW$ to approximate $M$ -- i.e. (approximately) minimize $\norm{M - AW}_F$
where $\norm{}_F$ denotes the Frobenius norm; we refer to this as  {\em Approximate NMF}. 

This problem has a rich history spanning quantum mechanics, probability theory, data analysis, polyhedral combinatorics, communication complexity, demography, chemometrics, etc. In the past decade NMF has become enormously popular in machine learning, where $A$ and $W$ are computed using a variety of local search heuristics. Vavasis recently proved that this problem is NP-complete. (Without the restriction that $A$ and $W$ be nonnegative, both the exact and approximate problems can be solved optimally via the singular value decomposition.) 

We initiate a study of when this problem is solvable in polynomial time.
Our results are the following:
\begin{enumerate}
\item We give a polynomial-time algorithm for exact and approximate NMF for every constant $r$.
Indeed NMF is most interesting in applications precisely when $r$ is small. 
\item We complement this with a hardness result, that if exact $NMF$ can be solved in time $(nm)^{o(r)}$, $3$-SAT has a sub-exponential time algorithm. This rules out substantial improvements to the above algorithm. 
\item We give an algorithm that runs in time polynomial in $n$, $m$ and $r$ under the {\em separablity} condition identified by Donoho and Stodden in 2003. The algorithm may be practical since it is simple and noise tolerant (under benign assumptions). Separability is believed to hold in many practical settings.
\end{enumerate}
To the best of our knowledge, this last result is the first example of a polynomial-time algorithm that provably works under a non-trivial condition on the input and we believe that this will be an interesting and important direction for future work. 

\end{abstract}

\setcounter{page}{0} \thispagestyle{empty}
\newpage

\section{Introduction}

In the {\em Nonnegative Matrix Factorization} (NMF) problem we are given an $n \times m$ matrix $M$ with nonnegative real entries (such a matrix will be henceforth called ``nonnegative'') and an integer $r > 0$. Our goal is to express $M$ as $A W$ where $A$ and $W$ are nonnegative matrices of size $n \times r$ and $r \times m$ respectively. We refer to $r$ as the {\em inner-dimension} of the factorization and the smallest value of $r$ for which there is such a factorization as the {\em nonnegative rank} of $M$. An equivalent formulation is that our goal is to write $M$ as the sum of $r$ nonnegative rank-one matrices.\footnotemark[1] We note that $r$ must be at least the rank of $M$ in order for such a factorization to exist. In some applications, it makes sense to instead ask for $AW$ to be a good approximation to $M$ in some suitable matrix norm. We refer to the problem of finding a nonnegative $A$ and $W$ of inner-dimension $r$ that (approximately) minimizes $\norm{M - AW}_F$ as {\em Approximate NMF}, where $\norm{}_F$ denotes the Frobenius norm. Without the restriction that $A$ and $W$ be nonnegative, the problem can be solved exactly via singular value decomposition \cite{GV}. 

\footnotetext[1]{It is a common misconception that  since the real rank is the maximum number of linearly independent columns, the nonnegative rank must be the size of the largest set of columns in which no column can be written as a nonnegative combination of the rest. This is false, and has been the source of many incorrect proofs demonstrating a gap between rank and nonnegative rank. A correct proof finally follows from the results of Fiorini et al~\cite{FRT}.}

NMF is a fundamental problem that has been independently introduced in a number of different contexts and applications. 
 Many interesting heuristics and local search algorithms (including the familiar Expectation Maximization or EM) have been proposed to find such factorizations. One compelling family of applications is data analysis, where a nonnegative factorization is computed in order to extract certain latent relationships in the data and has been applied to image segmentation \cite{LS99}, \cite{LS00} information retrieval \cite{Hof} and document clustering \cite{XLG}. NMF also has applications in fields such as chemometrics \cite{LS} (where the problem has a long history of study under the name {\em self modeling curve resolution}) and biology (e.g. in vision research \cite{BB}): in some cases, the underlying physical model for a system has natural restrictions that force a corresponding matrix factorization to be nonnegative. In demography  (see e.g., \cite{Hen}), NMF is used to model the dynamics of marriage through a mechanism similar to the chemical laws of mass action. In combinatorial optimization, Yannakakis \cite{Yan} characterized the number of extra variables needed to succinctly describe a given polytope as the nonnegative rank of an appropriate matrix (called the ``slack matrix''). In communication complexity, Aho et al \cite{AUY} showed that the log of the nonnegative rank of a Boolean matrix is polynomially related to its deterministic communication complexity - and hence the famous Log-Rank Conjecture of Lovasz and Saks \cite{LOG} is equivalent to showing a quasi-polynomial relationship  between real rank
and nonnegative rank for  Boolean matrices. In complexity theory, Nisan used nonnegative rank to prove lower bounds for non-commutative models of computation~\cite{Nisan}. Additionally, the 1993 paper of Cohen and Rothblum~\cite{CR93} gives a long list of other applications in statistics and quantum mechanics. That paper also gives an exact algorithm that runs in exponential time. 

\begin{question} \label{qs:1}
Can a nonnegative matrix factorization be computed efficiently when the inner-dimension, $r$, is small? 
\end{question}
Vavasis recently proved that the NMF problem is $NP$-hard when $r$ is large
\cite{Vav}, but this only rules out an algorithm whose running time is polynomial in $n$, $m$ and $r$. Arguably, in most significant applications, $r$ is small.
Usually the algorithm designer posits a two-level generative model for the data and uses NMF to compute ``hidden'' variables that explain the data. This explanation is only interesting when the number of hidden variables ($r$) is much smaller than the number of examples ($m$) or the number of observations per example ($n$).  In information retrieval, we often take $M$ to be a ``term-by-document'' matrix where the $(i, j)^{th}$ entry in $M$ is the frequency of occurrence of the $i^{th}$ term in the $j^{th}$ document in the database. In this context, a NMF computes $r$ ``topics'' which are each a distribution on words (corresponding to the $r$ columns of $A$) and each document (a column in $M$) can be expressed as a distribution on topics given by the corresponding column of $W$ \cite{Hof}. This example will be a useful metaphor for thinking about nonnegative factorization. In particular it justifies the assertion  $r$ should be small -- the number of topics should be much smaller than the total number of documents in order for this representation to be meaningful. See Section~\ref{sec:appendix:remarks} for more details.

Focusing on applications, and the overwhelming empirical evidence that heuristic algorithms do find good-enough factorizations in practice, motivates our next question.

\begin{question} \label{qs:2}
Can we design very efficient algorithms for NMF if we make reasonable assumptions about $M$? 
\end{question}

\subsection{Our Results}

Here we largely resolve Question~\ref{qs:1}. We give both an algorithm for accomplishing this algorithmic task that runs in polynomial time for any constant value of $r$ and we complement this with an intractability result which states that assuming the Exponential Time Hypothesis \cite{IP} no algorithm can solve the exact NMF problem in time $(nm)^{o(r)}$. 

\begin{theorem}
There is an algorithm for the Exact NMF problem (where $r$ is the target inner-dimension) that runs in time $O( (nm)^{r^2 2^r})$. 
\end{theorem}

This result is based on algorithms for deciding the first order theory of the reals - roughly the goal is to express the decision question of whether or not the matrix $M$ has nonnegative rank at most $r$ as a system of polynomial equations and then to apply algorithms in algebraic geometry to determine if this semi-algebraic set is non-empty. The complexity of these procedures is dominated by the number of distinct variables occurring in the system of polynomial equations.  In fact, the number of distinct variables plays an analogous role to VC-dimension, in a sense and the running time of algorithms for determining if a semi-algebraic set is non-empty depend exponentially on this quantity. Additionally these algorithms can compute successive approximations to a point in the set at the cost of an additional factor in the run time that is polynomial in the number of bits in the input and output. The naive formulation of the NMF decision problem as a non-emptiness problem is to use $nr + mr$ variables, one for each entry in $A$ or $W$~\cite{CR93}. This would be unacceptable, since even for constant values of $r$, the associated algorithm would run in time exponential in $n$ and $m$. 

At the heart of our algorithm is a structure theorem -- based on a novel method for reducing the number of variables needed to define the associated semi-algebraic set. We are able to express the decision problem for nonnegative matrix factorization using $r^2 2^r$ distinct variables (and we make use of tools in geometry, such as the notion of a separable partition, to accomplish this \cite{Har}, \cite{AO}, \cite{HR}). Thus we obtain the algorithm quoted in the above theorem. All that was known prior to our work (for constant values for $r$) was an exponential time algorithm, and local search heuristics akin to the Expectation-Maximization (EM) Algorithm with unproved correctness or running time.

A natural requirement on $A$ is that its columns be linearly independent.
In most applications, NMF is used to express a large number of observed variables using a small number of hidden variables. 
If the columns of $A$ are not linearly independent then Radon's Lemma implies that  this expression can be far from unique. In the example from information retrieval, this translates to: there are candidate documents that can be expressed as a convex combination of one set of topics, or could alternatively be expressed as a convex combination of an entirely disjoint set of topics (see Section~\ref{subsec:simplicialjustify}). 
When we add the requirement that the columns of $A$ be linearly independent, we refer to the associated problem as the {\em Simplicial Factorization} (SF) problem. In this case the doubly-exponential dependence on $r$ in the previous theorem can be improved to singly-exponential.
Our algorithm is again based on the first order theory of the reals, but here the system of equations is much smaller so in practice one may be able to use heuristic approaches to solve this system (in which case, the validity solution can be easily checked). 
\begin{theorem}
There is an algorithm for the Exact SF problem (where $r$ is the target inner-dimension) that runs in time $O( (nm)^{r^2 })$. 
\end{theorem}

We complement these algorithms with a fixed parameter intractability result. We make use of a recent result of Patrascu and Williams \cite{PW} (and engineer low-dimensional gadgets inspired by the gadgets of Vavasis \cite{Vav}) to show that under the Exponential Time Hypothesis \cite{IP}, there is no exact algorithm for NMF that runs in time $(nm)^{o(r)}$. 
 This intractability result holds also for the SF problem.

\begin{theorem}
If there is an exact algorithm for the  SF problem (or for the NMF problem) that runs in time $O((nm)^{o(r)})$ then $3$-SAT can be solved in $2^{o(n)}$ time on instances with $n$ variables. 
\end{theorem}


Now we turn to Question~\ref{qs:2}. We consider the nonnegative matrix factorization problem under the "separability" assumption introduced by Donoho and Stodden \cite{DS} in the context of image segmentation. Roughly, this assumption asserts that there are $r$ rows of $A$ that can be permuted to form the identity matrix. If we knew the names of these rows, then computing a nonnegative factorization would be easy. The challenge in this context, is to avoid brute-force search (which runs in time $n^r$) and to find these rows in time polynomial in $n$, $m$ {\em and} $r$. 
To the best of our knowledge the following is the first example of a polynomial-time algorithm that provably works under a non-trivial condition on the input. 

\begin{theorem}
There is an exact algorithm that can compute a separable, nonnegative factorization $M = AW$ (where $r$ is the inner-dimension) in time polynomial in $n$, $m$ and $r$ if such a factorization exists. 
\end{theorem}

Donoho and Stodden \cite{DS} argue that the separability condition is naturally met in the context of image segmentation.
Additionally, Donoho and Stodden prove that separability in conjunction with some other conditions guarantees that the solution to the NMF problem is unique. Our theorem above is an algorithmic counterpart to their results, but requires only  separability. Our algorithm can also be made noise tolerant, and hence works even when the separability condition only holds in an approximate sense. Indeed, an approximate separability condition is regarded as a fairly benign assumption and is believed to hold in many practical contexts in machine learning. For instance it is usually satisfied by model parameters fitted to various generative models (e.g. LDA \cite{LDA} in information retrieval). (We thank David Blei for this information.)

Lastly, we consider the case in which the given
matrix $M$ does not have an exact low-rank NMF but rather
can be approximated by a nonnegative factorization with small inner-dimension. 

\begin{theorem}
There is a  $2^{\poly(r\log(1/\epsilon))}\poly(n,m)$-time algorithm that, given a $M$ for which there is a nonnegative factorization $AW$ (of inner-dimension $r$) which is
an $\epsilon$-approximation to $M$ in Frobenius norm, computes $A'$ and $W'$ satisfying 
 $$\norm{M-A'W'}_F \le O(\epsilon^{1/2}r^{1/4})\norm{M}_F.$$ 
\end{theorem}

The rest of the paper is organized as follows: In Section~\ref{sec:simp} we give an exact algorithm for the SF problem and in Section~\ref{sec:gen} we give an exact algorithm for the general NMF problem. In Section~\ref{sec:hard} we prove a fixed parameter intractability result for the SF problem. And in Section~\ref{sec:sep} and Section~\ref{sec:adv} we give algorithms for the separable and adversarial nonnegative fatorization problems. Throughout this paper, we will use the notation that $M_i$ and $M^j$ are the $i^{th}$ column and $j^{th}$ row of $M$ respectively.

\section{Simplicial Factorization}

\label{sec:simp}
Here we consider the simplicial factorization problem, in which the target inner-dimension is $r$ and 
the matrix $M$ itself has rank $r$. Hence in any factorization $M = AW$ (where $r$ is the inner-dimension),
$A$ must have full column rank and $M$ must have full row rank. 

\subsection{Justification for Simplicial Factorization}
\label{subsec:simplicialjustify}
We first argue that the extra restriction imposed in simplicial factorization is natural in
many contexts: 
Through a re-scaling (see Section~\ref{sec:appendix:separable} for more details),
 we can assume that the
columns of $M$, $A$ and $W$ all have unit $\ell_1$ norm. The factorization $M =AW$  
can be interpreted probabilistically: each column
of $M$ can be expressed as a  convex combination (given by the corresponding column of $W$) of columns in $A$.
In the example in the introduction, columns of $M$ represent
documents and the columns
of $A$ represent ``topics''. Hence a nonnegative factorization is an ``explanation''
: each document can be expressed as a convex combination of the topics. 

But if $A$
does not have full column rank then this explanation is seriously
deficient. This follows from a restatement of Radon's Lemma. Let $conv(A_U)$ be the convex
hull of the columns $A_i$ for $i \in U$. 

\begin{observation}
If $A$ is an $n \times r$ (with $n \geq r$) matrix and $rank(A) < r$, then there are two disjoint sets of columns $U, V \subset [r]$ so that $conv(A_U) \cap conv(A_V) \neq \emptyset$.
\end{observation}

 The observation implies that there is some candidate document $x$ that can be expressed as a convex combination
 of topics (in $U$), or instead can be expressed as a convex combination of 
 an entirely disjoint set ($V$) of topics. 
 The end goal of NMF is often to use the representation of documents as distributions on topics
 to perform various tasks, such as clustering or information retrieval. But if (even given the set of topics
in a database) it is this ambiguous to determine how we should represent a given document as
a convex combination of topics, then the topics we have extracted cannot be very useful for clustering!
 In fact, it seems unnatural to not require the columns of $A$ to be linearly
independent!

Next, one should consider the process (probabilistic, presumably) that
generates the datapoints, namley, columns of $M$. Any reasonable
process for generating columns of $M$ from the 
columns of $A$ would almost surely result in a matrix $M$ whose rank
equals the rank of $A$. 
But then $M$ has the same rank as $A$. 
\vspace*{-0.1in}
\subsection{Algorithm for Simplicial Factorization}
In this Section we give an algorithm that solves the simplicial factorization problem in $(nm)^{O(r)}$ time.  Let $L$ be the maximum bit complexity of any 
coefficient in the input. 

\begin{theorem}
There is an $O((nm)^{O(r^2)})$ time algorithm for deciding if the simplicial
factorization problem has a solution of inner-dimension at most $r$. Furthermore, we can compute a rational approximation to the solution up
to accuracy $\delta$ in time $\poly(L, (nm)^{O(r^2)}, \log 1/\delta)$. 
\end{theorem}

The above theorem is proved by using
Lemma~\ref{lem:simplicialstruc} below to reduce the problem of finding a simplicial factorization to finding a
point inside a semi-algebraic set with $poly(n)$ constraints and $2 r^2$
real-valued variables (or deciding that this set is empty). 
The decision problem can be solved using the well-known algorithm of Basu et. al.\cite{BPR} solves this
problem in $n^{O(r^2)}$ time. We can instead use the algorithm of Renegar \cite{Ren} (a bound of $\poly(L, (nm)^{O(r^2)})$ on the bit complexity
of the coefficients in the solution due to Grigor'ev and Vorobjov \cite{GriVor}) to compute a rational approximation to the solution up
to accuracy $\delta$ in time $\poly(L, (nm)^{O(r^2)}, \log 1/\delta)$. 

This reduction
uses the fact that since $A, W$  have full rank they have ``pseudo-inverses'' 
 $A^+$,  $W^+$  which are $r \times n$ and $n\times r$ matrices respectively 
such that $A^+ A = W W^+ = I_{r\times r}$. 
Thus $A^+ M_i = A^+ A W_i = W_i$ and similarly $M^j W^+ = A^j$.


\begin{definition}
Let $C =\set{u_1, u_2, .. u_r }$ be a basis for the columns of $M$
in $\Re^n$, and let $R = \set{v_1, v_2, ... v_r} $ be a basis for the rows of $M$ in $\Re^m$.


Then $M_C$ (a size $r \times m$ matrix) denotes the columns of $M$
expressed in the basis $\cal{C}$, and similarly 
 $M_{R}$ (a size $n \times r$ matrix) denotes the rows of $M$ expressed in the basis $\cal{R}$. 
\end{definition}


\begin{lemma}[Structure Lemma for Simplicial Factorization]\label{lem:simplicialstruc}
$M$ has a simplicial factorization rank $r$  iff for every basis
$C$ for the columns and basis $B$ for the rows of $M$, there are
$r \times r$ matrices $T_C, T_R$ such that:
(i) $T_C M_C$ and $M_R T_R$ are nonnegative matrices
(ii)  $M_R T_R T_C M_C = M$
\end{lemma}
\begin{proof}
(``if'') Suppose the conditions in the theorem are met. Then set $A = M_R T_R$ and $W = T_C M_C$. These matrices are nonnegative and have size $n \times r$ and $r \times m$ respectively, and furthermore are a factorization for $M$. Since $rank(M) = r$, $A$ and $W$ are a simplicial factorization. 

(``only if'') Conversely suppose that there is a simplicial
factorization $M = AW$.  Let 
$\cal{C}= \set{u_1, u_2, .. u_r} $ and $\cal{R} = \set{v_1, v_2, ... v_r}$ be  
{\em arbitrary} bases for the columns and rows of $M$ respectively. Let $U$ and $V$ be the corresponding $n \times r$ and $m \times r$ matrices. Let $M_C$ and $M_R$ be $r \times m$ and $n \times r$ representations in this basis for the columns and rows of $M$ - i.e. $U M_C = M$ and $M_R V^T = M$. 

Define $r\times r$ matrices $T_C= A^+ U$ and $T_R = V^T W^+$ 
where $A^+$ and $W^+$ are the respective pseudoinverses of $A, W$.
Let us check that this choice of $T_C$ and $T_R$ satisfies the conditions in the theorem. 

We can re-write $T_C M_C = A^+ U M_C = A^+ M = W$ and hence the first condition in the theorem is satisfied. Similarly $M_R T_R = M_R V^T W^+ = M W^+ = A$ and hence the second and third condition are also satisfied. 
\end{proof}

\section{General NMF}\label{sec:gen}
Now we consider the NMF problem where the factor matrices $A, W$ need not have full rank. 

\begin{theorem}
There is a $O( (nm)^{cr^2 2^r})$ time deterministic algorithm that given an $n \times m$ nonnegative matrix $M$ outputs a
factorization $AW$ of inner dimension $r$  if such a factorization exists.
\end{theorem}

As in the Simplicial case the main idea will again be a reduction to an existence question for a semi-algebraic set, but this reduction is
significantly more complicated than Lemma~\ref{lem:simplicialstruc}.

\subsection{General Structure Theorem: Minimality}
Our goal is to re-cast nonnegative matrix factorization (for constant $r$) as a system of polynomial inequalities where the number of variables is constant, the maximum degree is constant and the number of constraints is polynomially bounded in $n$ and $m$. The main obstacle is that $A$ and $W$ are {\em large} - we cannot afford to introduce a new variable to represent each entry in these matrices. We will demonstrate there is always a "minimal" choice for $A$ and $W$ so that:

\begin{enumerate}

\item there is a collection of linear transformations $T_1, T_2, ... T_{g(r)}$ from the column-span of $M$ to $\Re^r$ and a choice function $\sigma_W: [m] \rightarrow [g(r)]$

\item and a collection of linear transformations $S_1, S_2, ... S_{g(r)}$ from the row-span of $M$ to $\Re^r$ and a choice function $\sigma_A: [n] \rightarrow [g(r)]$

\end{enumerate}

And these linear transformations and choice functions satisfy the conditions:

\begin{enumerate}

\item for each $i \in [n]$, $W_i = T_{\sigma_W(i)} M_i$ and

\item for each $j \in [m]$, $A^j = M^j S_{\sigma_A(j)} $.

\end{enumerate}

Furthermore, the number of possible choice functions $\sigma_W$ is at most $m^{cr^2 f(r)}$ and the number of possible choice functions for $\sigma_A$ is at most $n^{cr^2 g(r)}$. These choice functions are based on the notion of a simplicial partition, which we introduce later. We then give an algorithm for enumerating all simplicial partitions (this is the primary bottleneck in the algorithm). Fixing the choice functions $\sigma_W$ and $\sigma_A$, the question of finding linear transformations $T_1, T_2, ... T_{g(r)}$ and $S_1, S_2, ... S_{g(r)}$ that satisfy the above constraints (and the constraint that $M = AW$, and $A$ and $W$ are nonnegative) is exactly a system of polynomial inequalities with a $O(r^2 g(r))$ variables (each matrix $T_i$ or $S_j$ is $r \times r$), degree at most four and furthermore there are at most $O(mn)$ polynomial constraints.

In this subsection, we will give a procedure (which given $A$ and $W$) generates a "minimal" choice for $A$ and $W$ (call this minimal choice $A'$ and $W'$), and we will later establish that this "minimal" choice satisfies the structural property stated informally above.

\begin{definition}
Let $\calC(A) \subset 2^{[r]}$ denote the subsets of $[r]$ corresponding to maximal independent sets of columns (of $A$). Similarly let $\calR(W) \subset 2^{[r]}$ denote the subsets of $[r]$ corresponding to maximal independent sets of rows (of $W$). 
\end{definition}

A basic fact from linear algebra is that all maximal independent sets of columns of $A$ have exactly $rank(A)$ elements and all maximal independent sets of rows of $W$ similarly have exactly $rank(W)$ elements. 

\begin{definition}
Let $\succ_s$ be the total ordering on subsets of $[r]$ of size $s$ so that if $U$ and $V$ are both subsets of $[r]$ of size $s$, $U \prec_s V$ iff $U$ is lexicographically before $V$. 
\end{definition}

\begin{definition}
Given a column $M_i$, we will call a subset $U \in \calC(A)$ a minimal basis for $M_i$ (with respect to $A$) if $M_i \in cone(A_U)$ and for all $V \in \calC(A)$ such that $M_i \in cone(A_V)$ we must have $U \prec_s V$. 
\end{definition}

\begin{claim}~\label{claim:basis}
If $M_i \in cone(A)$, then there is some $U \in \calC(A)$ such that $M_i \in cone(A_U)$. 
\end{claim}





\begin{definition}

A proper chain $(A, W, A', W')$ is a set of nonnegative matrices for which $M = AW$, $M = AW'$ and $M = A'W'$ (the inner dimension of these factorizations is $r$) and functions $\sigma_{W'}: [m] \rightarrow \calC(A)$ and $\sigma_{A'}: [n] \rightarrow \calR(W')$ such that 

\begin{enumerate}

\item for all $i \in [m]$, $A W'_i = M_i$, $supp(W'_i) \subset \sigma_{W'}(i)$ and $\sigma_{W'}(i)$ is a minimal basis with respect to $A$ for $M_i$

\item for all $j \in [n]$, $A'_j W' = M^j$, $supp(A^j) \subset \sigma_{A'}(j)$ and $\sigma_{A'}(j)$ is a minimal basis with respect to $W'$ for $M^j$. 

\end{enumerate}

\end{definition}

Note that the extra conditions on $W'$ (i.e. the minimal basis constraint) is with respect to $A$ and the extra conditions on $A'$ are with respect to $W'$. This simplifies the proof that there is always some proper chain, since we can compute a $W'$ that satisfies the above conditions with respect to $A$ and then find an $A'$ that satisfies the conditions with respect to $W'$. 

\begin{lemma}~\label{lemma:chain}
If there is a nonnegative factorization $M = AW$ (of inner-dimension $r$), then there is a choice of nonnegative $A', W'$ of inner-dimension $r$ and functions $\sigma_{W'}: [m] \rightarrow \calC(A)$ and $\sigma_{A'}: [n] \rightarrow \calR(W')$ such that $(A, W, A', W')$ and $\sigma_{W'}$, $\sigma_{A'}$ form a proper chain.
\end{lemma}

\begin{proof}
The condition that there is some nonnegative $W$ for which $M = AW$ is just the condition that for all $i \in [m]$, $M_i \in cone(A)$. Hence, for each vector $M_i$, we can choose a minimal basis $U \in \calC(A)$ using Claim~\ref{claim:basis}. Then $M_i \in cone(A_U)$ so there is some nonnegative vector $W'_i$ supported on $U$ such $A W'_i = M_i$ and we can set $\sigma_{W'}(i) = U$. Repeating this procedure for each column $M_i$, results in a nonnegative matrix $W'$ that satisfies the condition $M = A W'$ and for each $i \in [m]$, by design $supp(W'_i) \subset \sigma_{W'}(i)$ and $\sigma_{W'}(i)$ is a minimal basis with respect to $A$ for $M_i$. 

We can re-use this argument above, setting $M^T = (W'^T) A^T$ and this interchanges the role of $A$ and $W$. Hence we obtain a nonnegative matrix $A'$ which satisfies $M = A' W'$ and for each $j \in [n]$, again by design we have that $supp(A^j) \subset \sigma_{A'}(j)$ and $\sigma_{A'}(j)$ is a minimal basis with respect to $W$ for $M^j$.
\end{proof}

\begin{definition}
Let $\Pi(A, U)$ (for $U \in \calC(A)$) denote the $r \times n$ linear transformation that is zero on all rows not in $U$ (i.e. $\Pi(A, U)^j = \vec{0}$ for $j \notin U$) and restricted to $U$ is $\Pi(A, U)^U = (A_U)^+$ (where the $+$ operation denotes the Moore-Penrose pseudoinverse). 
\end{definition}

\begin{lemma}~\label{lemma:reconstr}
Let $(A, W, A', W')$ and $\sigma_{W'}$ and $\sigma_{A'}$ form a proper chain. For any index $i \in [m]$, let $U_i = \sigma_{W'}(i)$ and for any index $j \in [n]$ let $V_j = \sigma_{A'}(j)$.  Then $W'_i = \Pi(A, U_i) M_i$ and $A'^j = M^j \Pi(W'^T, V_j)^T$.
\end{lemma}

Notice that in the above lemma, the linear transformation that recovers the columns of $W'$ is based on column subsets of $A$, while the linear transformation to recover the rows of $A'$ is based on the row subsets of $W'$ (not $W$). 

\begin{proof}
Since $(A, W, A', W')$ and $\sigma_{W'}$ and $\sigma_{A'}$ form a proper chain we have that $AW' = M$. Also $supp(W'_i) \subset U_i = \sigma_{W'}(i)$. Consider the quantity $ \Pi(A, U_i) M_i$. For any $j \notin U_i$, $ (\Pi(A, U_i) M_i)_j = 0$. So consider $$(\Pi(A, U_i) M_i)_{U_i} =  (A_{U_i})^+ A W'_i = (A_{U_i})^+ A_{U_i} (W'_i)_{U_i}$$ where the last equality follows from the condition $supp(W'_i) \subset U_i $. Since $U_i \in \calC(A)$ we have that $(A_{U_i})^+ A_{U_i}$ is the $|U_i| \times |U_i|$ identity matrix. Hence $W'_i = \Pi(A, U_i) M_i$. An identical argument with $W'$ replaced with $A'$ and with $A$ replaced by $W'^T$ (and $i$ and $U_i$ replaced with $j$ and $V_j$) respectively implies that $A'^j = M^j \Pi(W'^T, V_j)^T$ too. 
\end{proof}

Note that there are at most $|\calC(A)| \leq 2^r$ linear trasformations of the form $\Pi(A, U_i)$ and hence the columns of $W'$ can be recovered by a constant number of linear transformations of the column span of $M$, and similarly the rows of $A'$ can also be recovered. 

The remaining technical issue is we need to demonstrate that there are not too many (only polynomially many, for constant $r$)  choice functions $\sigma_{W'}$ and $\sigma_{A'}$ and that we can enumerate over this set efficiently. In principle, even if say $\calC(A)$ is just two sets, there are exponentially many choices of which (of the two) linear transformation to use for each column of $M$. However, when we use lexicographic ordering to tie break (as in the definition of a minimal basis), the number of choice functions is polynomially bounded. We will demonstrate that the choice function $\sigma_{W'}: [m] \rightarrow \calC(A)$ arising in the definition of a proper chain can be embedded in a restricted type of geometric partitioning of $M$ which we call a simplicial partition.

\subsection{General Structure Theorem: Simplicial Partitions}

Here, we establish that the choice functions $\sigma_{W'}$ and $\sigma_{A'}$ in a proper chain are combinatorially simple. The choice function $\sigma_{W'}$ can be regarded as a partition of the columns of $M$ into $|\calC(A)|$ sets, and similarly the choice function $\sigma_{A'}$ is a partition of the rows of $M$ into $\calR(W')$ sets. Here we define a geometric type of partitioning scheme which we call a simplicial partition, which has the property that there are not too many simplicial partitions (by virtue of this class having small VC-dimension), and we show that the partition functions $\sigma_{W'}$ and $\sigma_{A'}$ arising in the definition of a proper chain are realizable as (small) simplicial partitions. 

\begin{definition}
A $(k, s)$-simplicial partition of the columns of $M$ is generated by a collection of $k$ sets of $s$ hyperplanes $$\calH^1 = \{h_1^1, h_2^1, ... h_s^1\},  \calH^2 = \{h_1^2, h_2^2, ... h_s^2\}, ... \calH^k = \{h_1^k, h_2^k, ... h_s^k\}.$$ Let $Q_i = \{i' \mbox{ s.t. for all } j \in [s], h_j^{i} \cdot M_{i'} \geq 0 \}$. Then this collection of sets of hyperplanes results in the partition 
\begin{itemize}

\item $P_1 = Q_1$

\item $P_2 = Q_2 - P_1$ 

\item $P_k = Q_k - P_1 - P_2... -P_{k-1}$

\item $P_{k+1} = [m] - P_1 - P_2 ... -P_k$
\end{itemize}
\end{definition}

If $rank(A) = s$, we will be interested in a $( {r \choose s}, s)$-simplicial partition. 

\begin{lemma}~\label{lemma:geom}
Let $(A, W, A', W')$ and $\sigma_{W'}$ and $\sigma_{A'}$ form a proper chain. Then the partitions corresponding to $\sigma_{W'}$ and to $\sigma_{A'}$ (of columns and rows of $M$ respectively) are a $( {r \choose s}, s)$-simplicial partition and a $( {r \choose t}, t)$-simplicial partition respectively, where $rank(A) = s$ and $rank(W') = t$. 
\end{lemma}

\begin{proof}
Order the sets in $\calC(A)$ according to the lexicographic ordering $\succ_s$, so that $V_1 \prec_s V_2 \prec_s ... V_k$ for $k = |\calC(A)|$. Then for each $j$, let $\calH^j$ be the rows of the matrix $(A_{V_j})^+$. Note that there are exactly $rank(A) = s$ rows, hence this defines a $(k, s)$-simplicial partition. 

\begin{claim}
$\sigma_{W'}(i) = j$ if and only if $M_i \in P_j$ in the $(k, s)$-simplicial partition generated by $\calH^1, \calH^2, ... \calH^k$. 
\end{claim} 

\begin{proof}
Since $(A, W, A', W')$ and $\sigma_{W'}$ and $\sigma_{A'}$ forms a proper chain, we have that $M = AW'$. Consider a column $i$ and the corresponding set $V_i = \sigma_{W'}(i)$. Recall that $V_j$ is the $j^{th}$ set in $\calC(A)$ according to the lexicographic ordering $\succ_s$. Also from the definition of a proper chain $V_i$ is a minimal basis for $M_i$ with respect to $A$. Consider any set $V_{j'} \in \calC(A)$ with $j' < j$. Then from the definition of a minimal basis we must have that $M_i \notin cone(A_{V_{j'}})$. Since $V_{j'} \in \calC(A)$, we have that the transformation $(A_{V_{j'}}) (A_{V_{j'}})^+$ is a projection onto $span(A)$ which contains $span(M)$. Hence $(A_{V_{j'}}) (A_{V_{j'}})^+ M_i = M_i$, but $M_i \notin cone(A_{V_{j'}})$ so $(A_{V_{j'}})^+ M_i$ cannot be a nonnegative vector. Hence $M_i$ is not in $P_{j'}$ for any $j' < j$. Furthermore, $M_i$ is in $Q_j$: using Lemma~\ref{lemma:reconstr} we have $\Pi(A, V_j) M_i = \Pi(A, V_j) A W'_i = W'_i \geq \vec{0}$ and so $(A_{V_j})^+ M_i = (\Pi(A, V_j) M_i )_{V_j} \geq \vec{0}$. 
\end{proof}

We can repeat the above replacing $A$ with $W'^T$ and $W'$ with $A'$, and this implies the lemma. 
\end{proof}

\subsection{Enumerating Simplicial Partitions}

Here we give an algorithm for enumerating all $(k, s)$-simplicial partitions (of, say, the columns of $M$) that runs in time $O(m^{ks(r+1)})$. An important observation is that the problem of enumerating all simplicial partitions can be reduced to enumerating all partitions that arise from a single hyperplane. Indeed, we can over-specify a simplicial partition by specifying the partition (of the columns of $M$) that results from each hyperplane in the set of $ks$ total hyperplanes that generates the simplicial partition. From this set of partitions, we can recover exactly the simplicial partition. 

A number of results are known in this domain, but surprisingly we are not aware of any algorithm that enumerates all partitions of the columns of $M$ (by a single hyperplane) that runs in polynomial time (for $dim(M) \leq r$ and $r$ is constant) without some assumption on $M$. For example, the VC-dimension of a hyperplane in $r$ dimensions is $r + 1$ and hence the Sauer-Shelah lemma implies that there are at most $O(m^{r+1})$ distinct partitions of the columns of $M$ by a hyperplane. In fact, a classic result of Harding (1967) gives a tight upper bound of $O(m^r)$. Yet these bounds do not yield an algorithm for efficiently enumerating this structured set of partitions without checking {\em all} partitions of the data.

A recent result of Hwang and Rothblum \cite{HR} comes close to our intended application. A separable partition into $p$ parts is a partition of the columns of $M$ into $p$ sets so that the convex hulls of these sets are disjoint. Setting $p=2$, the number of separable partitions is exactly the number of distinct hyperplane partitions. Under the condition that $M$ is in general position (i.e. there are no $t$ columns of $M$ lying on a dimension $t-2$ subspace where $t = rank(M)-1$), Hwang and Rothblum give an algorithm for efficiently enumerating all distinct hyperplane partitions \cite{HR}. 

Here we give an improvement on this line of work, by removing any conditions on $M$ (although our algorithm will be slightly slower). The idea is to encode each hyperplane partition by a choice of not too many data points. To do this, we will define a slight generalization of a hyperplane partition that we will call a hyperplane separation:

\begin{definition}
A hyperplane $h$ defines a mapping (which we call a hyperplane separation) from columns of $M$ to $\{-1, 0, 1\}$ depending on the sign of $h \cdot M_i$ (where the sign function is $1$ for positive values, $-1$ for negative values and $0$ for zero). 
\end{definition}

A hyperplane partition can be regarded as a mapping from columns of $M$ to $\{-1, 1\}$ where we adopt the convention that $M_i$ such that $h \circ M_i$ is mapped to $1$. 

\begin{definition}
A hyperplane partition (defined by $h$) is an extension of a hyperplane separation (defined by $g$) if for all $i$, $g(M_i) \neq 0 \Rightarrow g(M_i) = h(M_i)$.
\end{definition}

\begin{lemma}~\label{lemma:extension}
Let $rank(M) = s$, then for any hyperplane partition (defined by $h$), there is a hyperplane $g$ that contains $s$ affinely independent columns of $M$ and for which $h$ (as a partition) is an extension of $g$ (as a separation). 
\end{lemma}

\begin{proof}
After an appropriate linear transformation (of the columns of $M$ and the hyperplanes), we can assume that $M$ is full rank. If the $h$ already contains $s$ affinely independent columns of $M$, then we can choose $g = h$. If not we can perturb $h$ in some direction so that for any column with $h(M_i) = 0$, we maintain the invariant that $M_i$ is contained on the perturbed hyperplane $h'$. Since $rank(M) = s$ this perturbation has non-zero inner product with some column in $M$ and so this hyperplane $h'$ will eventually contain a new column from $M$ (without changing the sign of $h(M_i)$ for any other column). We can continue this argument until the hyperplane contains $s$ affinely independent columns of $M$ and by design on all remaining columns agrees in sign with $h$. 
\end{proof}

\begin{lemma}~\label{lemma:char}
Let $rank(M) = s$. For any hyperplane $h$ (which defines a partition), there is a collection of $k \leq s$ sets of (at most $s $) columns of $M$, $S_1, S_2, .. S_k$ so that any hyperplanes $g_1, g_2, .. g_k$ which contain $S_1, S_2, ... S_k$ respectively satisfy: For all $i$, $h(M_i)$ (as a partition) is equal to the value of $g_j(M_i)$, where $j$ is the smallest index for which $g_j(M_i) \neq 0$. Furthermore these subsets are nested: $S_1 \supset S_2 \supset ... \supset S_k$. 
\end{lemma}

\begin{proof}
We can apply Lemma~\ref{lemma:extension} repeatedly. When we initially apply the lemma, we obtain a hyperplane $g_1$ that can be extended (as a separation) to the partition corresponding to $h$. In the above function (defined implicitly in the lemma) this fixes the partition of the columns except those contained in $g_1$. So we can then choose $M'$ to be the columns of $M$ that are contained in $g_1$, and recurse. If $S_2$ is the largest set of columns output from the recursive call, we can add columns of $M$ contained in $g_1$ to this set until we obtain a set of $s+1$ affinely independent columns contained in $g_1$, and we can output this set (as $S_1$). 
\end{proof}

\begin{theorem}~\label{thm:enumerate}
Let $rank(M) =s$. There is an algorithm that runs in time $O(m^{s}(s + 2)^s)$ time to enumerate all hyperplane partitions of the columns of $M$.
\end{theorem}

\begin{proof}
We can apply Lemma~\ref{lemma:char} and instead enumerate the sets of points $S_1, S_2, ... S_s$. Since these sets are nested, we can enumerate all choices as follows:

\begin{itemize}

\item choose at most $s$ columns corresponding to the set $S_1$

\item initialize an active set $T = S_1$

\item until $T$ is empty either

\begin{itemize}

\item choose a column to be removed from the active set

\item or indicate that the current active set represents the next set $S_i$ and choose the sign of the corresponding hyperplane

\end{itemize}

\end{itemize}

There are at most $O(m^{s}(s + 2)^s)$ such choices, and for each choice we can then run a linear program to determine if there is a corresponding hyperplane partition. (In fact, all partitions that result from the above procedure will indeed correspond to a hyperplane partition). The correctness of this algorithm follows from Lemma~\ref{lemma:char}.
\end{proof}

This immediately implies:

\begin{corollary}
There is an algorithm that runs in time $O(m^{k s^2)})$ that enumerates a set of partitions of the columns of $M$ that contains the set of all $(k, s)$-simplicial partitions (of the columns of $M$). 
\end{corollary}

\subsection{Solving Systems of Polynomial Inequalities}

The results of Basu et al \cite{BPR} give an algorithm for finding a point in a semi-algebraic set defined by $O(mn)$ constraints on polynomials of total degree at most $d$, and $f(r)$ variables in time $O( (mnd)^{cf(r)})$. Using our structure theorem for nonnegative matrix factorization, we will re-cast the decision problem of whether a nonnegative matrix $M$ has nonnegative rank $r$ as an existence question for a semi-algebraic set. 

\begin{theorem}
There is an algorithm for deciding if a $n \times m$ nonnegative matrix $M$ has nonnegative rank $r$ that runs in time $O( (nm)^{O(r^2 2^r)} )$. Furthermore, we can compute a rational approximation to the solution up
to accuracy $\delta$ in time $\poly(L, (nm)^{O(r^2 2^r)}, \log 1/\delta)$. 
\end{theorem}

\noindent We first prove the first part of this theorem using the algorithm of Basu et al \cite{BPR}, and we instead use the algorithm of Renegar \cite{Ren} to compute a rational approximation to the solution up to accuracy $\delta$ in time $\poly(L, (nm)^{O(r^2 2^r)}, \log 1/\delta)$. 

\vspace{0.5pc}

\begin{proof}
Suppose there is such a factorization. Using Lemma~\ref{lemma:chain}, there is also a proper chain. We can apply Lemma~\ref{lemma:geom} and using the algorithm in Theorem~\ref{thm:enumerate} we can enumerate over a superset of simplicial partitions. Hence, at least one of those partitions will result in the choice functions $\sigma_{W'}$ and $\sigma_{A'}$ in the proper chain decomposition for $M = AW$. 

Using Lemma~\ref{lemma:reconstr} there is a set of at most $2^r$ linear transformations $T_1, T_2, ... T_{2^r}$ which recover columns of $W'$ given columns of $M$, and similarly there is a set of at most $2^r$ linear transformations $S_1, S_2, ... S_{2^r}$  which recover the rows of $A'$ given rows of $M$. Note that these linear transformations are from the column-span and row-span of $M$ respectively, and hence are from subspaces of dimension at most $r$. So apply a linear transformation to columns of $M$ and one to rows of $M$ to to recover matrices $M_C$ and $M_R$ respectively (which are no longer necessarily nonnegative) but which are dimension $r \times m$ and $n \times r$ respectively. There will still be a collection of at most $2^r$ linear transformations from columns of $M_C$ to columns of $W'$, and similarly for $M_R$ and $A'$. 

We will choose $r^2$ variables for each linear transformation, so there are $2 *r^2 * 2^r$ variables in total. Then we can write a set of $m$ linear constraints to enforce that for each column of $(M_C)_i$, the transformation corresponding to $\sigma_{W'}(i)$ recovers a nonnegative vector. Similarly we can define a set of $n$ constraints based on rows in $M_R$. 

Lastly we can define a set of constraints that enforce that we do recover a factorization for $M$: For all $i \in [m], j \in [n]$, let $i' = \sigma_{W'}(i)$ and $j' = \sigma_{A'}(j)$. Then we write the constraint $(M_C)^j S_{j'} T_{i'}(M_R)_i = M_i^j$. This constraint has degree at two in the variables corresponding to the linear transformations. Lemma~\ref{lemma:chain} implies that there is some choice of these transformations that will satisfy these constraints (when we formulate these constraints using the correct choice functions in the proper chain decomposition). Furthermore, any set of transformations that satisfies these constraints does define a nonnegative matrix factorization of inner dimension $r$ for $M$. 

And of course, if there is no inner dimension $r$ nonnegative factorization, then all calls to the algorithm of Basu et al \cite{BPR} will fail and we can return that there is no such factorization. 
\end{proof}

The result in Basu et. al.~\cite{BPR} is a quantifier elimination algorithm in the Blum, Shub and Smale (BSS) model of computation~\cite{BCSS}. The BSS model is a model for real number computation and it is natural to ask what is the bit complexity of finding a rational approximation of the solutions. There has been a long line of research on the decision problem for first order theory of reals: given a quantified predicate over polynomial inequalities of reals, determine whether it is true or false. What we need for our algorithm is actually a special case of this problem: given a set of polynomial inequalities over real variables, determine whether there exists a set of values for the variables so that all polynomial inequalities are satisfied. In particular, all variables in our problem are quantified by existential quantifier and there are no alternations. For this kind of problem Grigor'ev and Vorobjov \cite{GriVor} first gave a singly-exponential time algorithm that runs in $(nd)^{O(f(r)^2)}$ where $n$ is the number of polynomial inequalities, $d$ is the maximum degree of the polynomials and $f(r)$ is the number of variables. The bit complexity of the algorithm is $\poly(L, (nd)^{O(f(r)^2)})$ where $L$ is the maximum length of the coefficients in the input. Moreover, their algorithm also gives an upperbound of $\poly(L, (nd)^{O(f(r))})$ on the number of bits required to represent the solutions. Renegar\cite{Ren} gave a better algorithm that for the special case we are interested in takes time $(nd)^{O(f(r))}$. Using his algorithm with binary search (with search range bounded by Grigor'ev et.al.\cite{GriVor}),  we can find rational approximations to the solutions with accuracy up to $\delta$ in time $\poly(L, (nm)^{O(f(r))}, \log 1/\delta)$.

We note that our results on the SF problem are actually a special case of the theorem above (because our structural lemma for simplicial factorization is a special case of our general structure theorem):

\begin{corollary}
There is an algorithm for determining whether the positive rank of a nonnegative $n \times m$ matrix $M$ equals the rank and this algorithm runs in time $O( (nm)^{cr^2})$.
\end{corollary}

\begin{proof}
If $rank(M) = r$, then we know that both $A$ and $W$ must be full rank. Hence $\calC(A)$ and $\calR(W)$ are both just the set $\{1, 2, ... r\}$. Hence we can circumvent the simplicial partition machinery, and set up a system of polynomial constraints in at most $2r^2$ variables. 
\end{proof}

\section{Strong Intractability of Simplicial Factorization}\label{sec:hard}
\label{sec:hardness}

Here we give evidence that finding a simplicial factorization of dimension $r$ probably cannot be solved in $(nm)^{o(r)}$ time, unless $3$-SAT can be solved in 
$2^{o(n)}$ time (in other words, if the Exponential Time Hypothesis of~\cite{IP} is true). Surprisingly, even the $NP$-hardness of the problem for general $r$ 
was only proved quite recently by Vavasis~\cite{Vav}. That reduction is the inspiration for our result, though unfortunately we were unable to use it directly to get 
low-dimensional instances. Instead we give a new  reduction using the $d$-SUM Problem. 

\begin{definition}[$d$-SUM]
In the $d$-SUM problem we are given a set of $N$ values $\{s_1, s_2, ... s_N\}$ each in the range $[0, 1]$, and the goal is to determine if there is a set of $d$ numbers (not necessarily distinct) that sum to exactly $d/2$. 
\end{definition}

This definition for the $d$-SUM Problem is slightly unconventional in that here we allow repetition (i.e. the choice of $d$ numbers need not be distinct). Patrascu and Williams \cite{PW} recently proved that if $d$-SUM can be solved in $N^{o(d)}$ time then $3$-SAT has a sub-exponential time algorithm. In fact, in the instances constructed in \cite{PW} we can allow repetition of numbers without affecting the reduction since in these instances choosing any number more than once will never result in a sum that is exactly $d/2$. Hence we can re-state the results in \cite{PW} for our (slightly unconventional definition for) $d$-SUM. 

\begin{theorem}
If $d < N^{0.99}$ and if $d$-SUM instances of $N$ distinct numbers each of $O(d \log N)$ bits can be solved in $N^{o(d)}$ time then $3$-SAT on $n$ variables can be solved in time $2^{o(n)}$. 
\end{theorem}

Given an instance of the $d$-SUM, we will reduce to an instance of the Intermediate Simplex problem defined in \cite{Vav}. 

\begin{definition}[Intermediate Simplex]
Given a polyhedron $P= \{x \in \Re^{r-1}: H x \ge b\}$ where $H$ is an $n \times (r-1)$ size matrix and $b \in \Re^n$ such that the matrix $[H, b]$ has rank $r$ and a set $S$ of $m$ points in $\Re^{r-1}$, the goal of the Intermediate Simplex Problem is to find a set of points $T$ that form a simplex (i.e. $T$ is a set of $r$ affinely independent points) each in $P$ such that the convex hull of $T$ contains the points in $S$. 
\end{definition}

Vavasis~\cite{Vav} proved that Intermediate Simplex is equivalent to the Simplicial Factorization problem.



\begin{theorem}[Vavasis, 2009~\cite{Vav}]
There is a polynomial time reduction from Intermediate Simplex problem to Simplicial Factorization problem and vice versa and furthermore both reductions preserve the value of $r$. 
\end{theorem}

Interestingly, an immediate consequence of this theorem is that Simplicial Factorization is easy in the case in which $rank(M) = 2$ because mapping these instances to instances of intermediate simplex results in a one dimensional problem - i.e. the polyhedron $P$ is an interval.

\subsection{The Gadget}\label{sec:gad}

Given the universe $U = \{s_1,s_2, \ldots, s_N\}$ for the $d$-SUM problem, we construct a two dimensional Intermediate Simplex instance as shown in Figure~\ref{fig:gadget}. We will show that the Intermediate Simplex instance has exactly $N$ solutions, each representing a choice of $s_i$. Later in the reduction we use $d$ such gadgets to represent the choice of $d$ numbers in the set $U$.

\begin{figure}
  \center
  \includegraphics[width=5in]{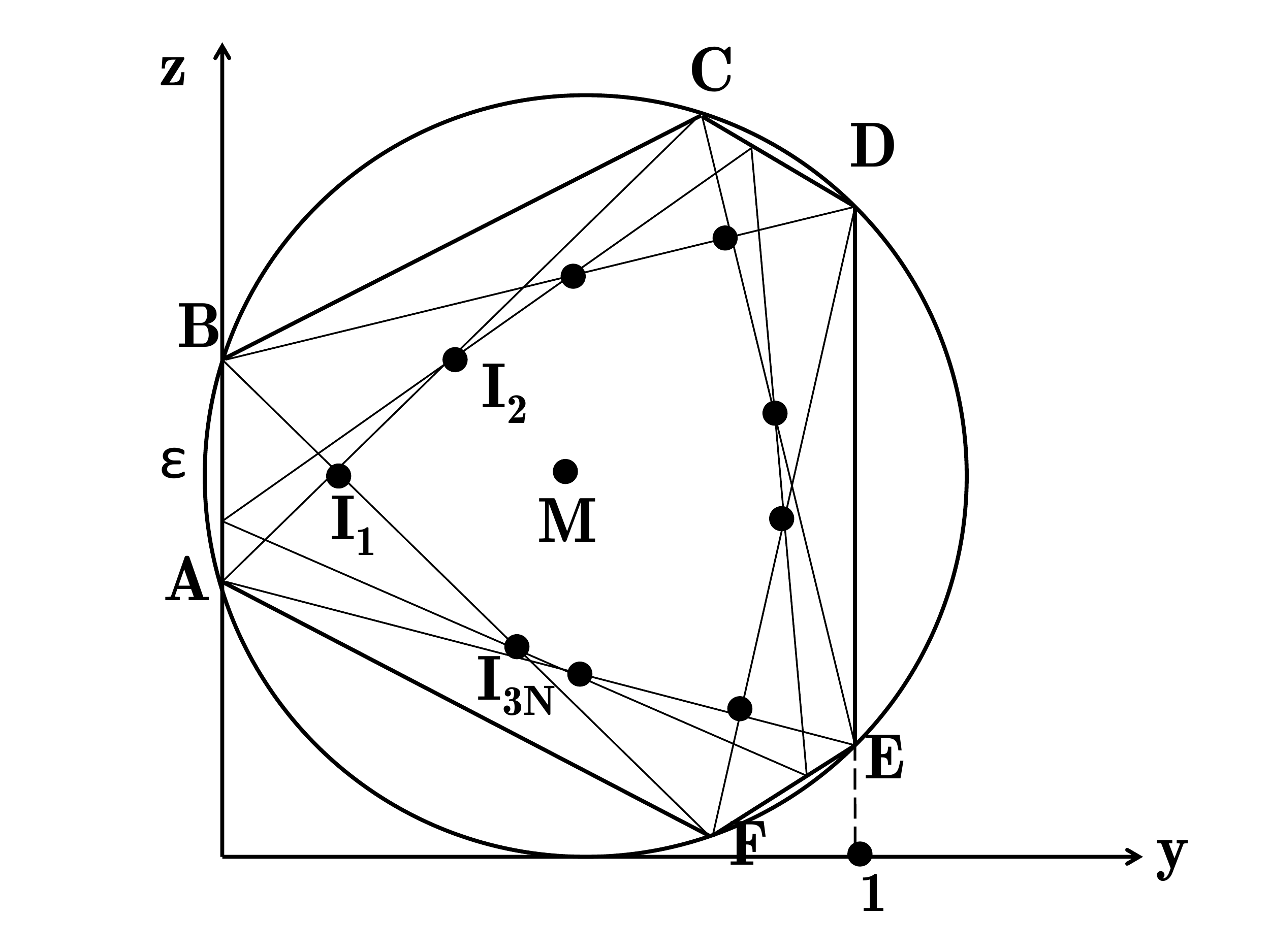}\\
\caption{The Gadget}
  \label{fig:gadget}
\end{figure}

Recall for a two dimensional Intermediate Simplex problem, the input consists of a polygon $\mathcal{P}$ (which is the hexagon $ABCDEF$ in Figure~\ref{fig:gadget}) and a set of points $S = \{I_1,I_2,\ldots, I_{3N}\}$ inside $\mathcal{P}$ (which are the  dots, except for $M$). A solution to this two dimensional Intermediate Simplex instance will be a triangle inside $\mathcal{P}$ such that all the points in $S$ are contained in the triangle (in Figure~\ref{fig:gadget} $ACE$ is a valid solution).

We first specify the polygon $\mathcal{P}$ for the Intermediate Simplex instance. The polygon $\mathcal{P}$ is just the hexagon $ABCDEF$ inscribed in a circle with center $M$. All angles in the hexagon are $2\pi/3$, the edges $AB=CD=EF = \epsilon$ where $\epsilon$ is a small constant depending on $N$, $d$ that we determine later. The other 3 edges also have equal lengths $BC=DE=FA$.

We use $y(A)$ and $z(A)$ to denote the $y$ and $z$ coordinates for the point $A$ (and similarly for all other points in the gadget). The hexagon is placed so that $y(A) = y(B) = 0$, $y(D) = y(E) = 1$. 

Now we specify the set $S$ of $3N$ points for the Intermediate Simplex instance. To get these points first take $N$ points in each of the 3 segements $AB$, $CD$, $EF$. On $AB$ these $N$ points are called $A_1$, $A_2$, ..., $A_N$, and $|AA_i| = \epsilon s_i$. Similarly we have points $C_i$'s on $CD$ and $E_i$'s on $EF$, $|CC_i| = |EE_i|  = \epsilon s_i$. Now we have $N$ triangles $A_iC_iE_i$ (the thin lines in Figure~\ref{fig:gadget}). We claim (see Lemma~\ref{lem:gadgetpoint} below) that the intersection of these triangles is a polygon with $3N$ vertices. The points in $S$ are just the vertices of this intersection. 

\begin{lemma}
\label{lem:gadgetpoint}
When $\epsilon < 1/50$, the points $\{A_i\}$, $\{C_i\}$, $\{E_i\}$ are on $AB$, $CD$, $EF$ respectively and $AA_i = CC_i = EE_i = \epsilon s_i$, the intersection of the $N$ triangles $\{A_iC_iE_i\}$ is a polygon with $3N$ vertices.
\end{lemma}

\begin{proof}
Since the intersection of $N$ triangles $A_iC_iE_i$ is the intersection of $3N$ halfplanes, it has at most $3N$ vertices. Therefore we only need to prove every edge in the triangles has a segment remaining in the intersection.
Notice that the gadget is symmetric with respect to rotations of $2\pi/3$ around the center $M$. 
 By symmetry we only need to look at edges $A_iC_i$. The situation here is illustrated in Figure~\ref{fig:gadgetproof}.

Since all the halfplanes that come from triangles $A_iC_iE_i$ contain the center $M$, later when talking about halfplanes we will only specify the boundary line. For example, the halfplane with boundary $A_iC_i$ and contains $E_i$ (as well as $M$) is called halfplane $A_iC_i$.

\begin{figure}
  \center
  \includegraphics[width=4in]{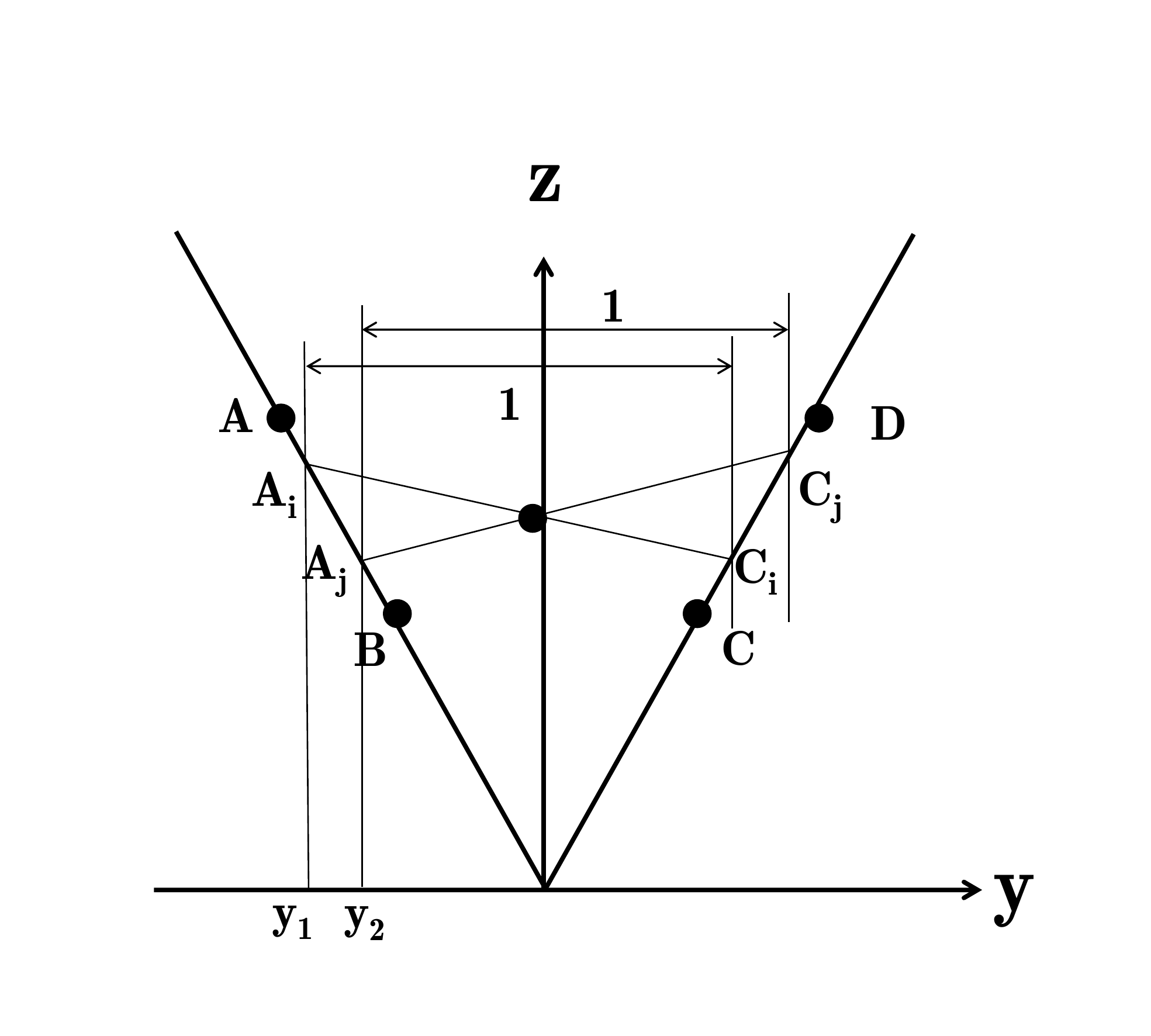}\\
\caption{Proof of Lemma~\ref{lem:gadgetpoint}}
  \label{fig:gadgetproof}
\end{figure}

The two thick lines in Figure~\ref{fig:gadgetproof} are extensions of $AB$ and $CD$, now they are rotated so that they are $z = \pm \sqrt{3} y$. The two thin lines are two possible lines $A_iC_i$ and $A_jC_j$. The differences between $y$ coordinates of $A_i$ and $C_i$ are the same for all $i$ (here normalized to 1) by the construction of the points $A_i$'s and $C_i$'s. Assume the coordinates for $A_i$, $A_j$ are $(y_i,- \sqrt{3} y_i)$ and $(y_j, -\sqrt{3} y_j)$ respectively. Then the coordinates for the intersection is $(y_i+y_j+1, \sqrt{3}(1+y_i+y_j+2y_iy_j))$. This means if we have $N$ segments with $y_1<y_2<\ldots < y_N$, segment $i$ will be the highest one when $y$ is in range $(y_{i-1}+y_i+1, y_i+y_{i+1}+1)$ (indeed, the lines with $j>i$ have higher slope and will win when $y>y_i+y_j+1\ge y_i+y_{i+1}+1$; the lines with $j<i$ have lower slope and will win when $y<y_i+y_j+1 \le y_i+y_{i-1}+1$). 

We also want to make sure that all these intersection points are inside the halfplanes $C_iE_i$'s and $E_iA_i$'s. Since $\epsilon < 1/50$, all the $y_i$'s are within $[-1/2-1/20, -1/2+1/20]$. Hence the intersection point is always close to the point $(0,\sqrt{3}/2)$, the distance is at most $1/5$. At the same time, since $\epsilon$ is small, the distances of this point $(0,\sqrt{3}/2)$ to all the $C_iE_i$'s and $E_iA_i$'s are all larger than $1/4$. 
Therefore all the intersection points are inside the other $2N$ halfplanes and the segments will indeed remain in the intersection. The intersection has $3N$ edges and $3N$ vertices.
\end{proof}

The Intermediate Simplex instance has $N$ obvious solutions: the triangles $A_iC_iE_i$, each one corresponds to a value $s_i$ for the $d$-SUM problem. In the following Lemma we show that these are the only possible solutions.

\begin{lemma}
\label{lem:gadget}
When $\epsilon < 1/1000$, if the solution of the Intermediate Simplex problem is $PQR$, then $PQR$ must be one of the $A_iC_iE_i$'s.
\end{lemma}

\begin{proof}
Suppose $PQR$ is a solution of the Intermediate Simplex problem, since $M$ is in the convex hull of $\{I_1,I_2,\ldots, I_{3N}\}$, it must be in $PQR$. Thus one of the angles $\angle PMQ$, $\angle QMR$, $\angle RMP$ must be at least $2\pi/3$ (their sum is $2\pi$). Without loss of generality we assume this angle is $\angle PMQ$ and by symmetry assume $P$ is either on $AB$ or $BC$. We shall show in either of the two cases, when $P$ is not one of the $A_i$'s, there will be some $I_k$ that is not in the halfplane $PQ$ (recall the halfplanes we are interested in always contain $M$ so we don't specify the direction).

When $P$ is on $AB$, since $\angle PMQ \ge 2\pi/3$, we have $CQ \ge AP$ (by symmetry when $CQ = AP$ the angle is exactly $2\pi/3$). This means we can move $Q$ to $Q'$ such that $CQ' = AP$. The intersection of halfplane $PQ'$  and the hexagon $ABCDEF$ is at least as large as the intersection of halfplane $PQ$ and the hexagon. However, if $P$ is not any of the points $\{A_i\}$ (that is, $|PQ'|/\epsilon \not\in \{s_1,s_2,...,s_N\}$), then $PQ'$ can be viewed as $A_{N+1}C_{N+1}$ if we add $s_{N+1} = |AP|/\epsilon$ to the set $U$. By Lemma~\ref{lem:gadgetpoint} introducing $PQ'$ must increase the number of vertices. One of the original vertices $I_k$ is not in the hyperplane $PQ'$, and hence not in $PQR$. Therefore when $P$ is on $AB$ it must coincide with one of the $A_i$'s, by symmetry $PQR$ must be one of $A_iC_iE_i$'s.

When $P$ is on $BC$, there are two cases as shown in Figure~\ref{fig:gadgetproof2}.

\begin{figure}
  \center
  \includegraphics[height=3in]{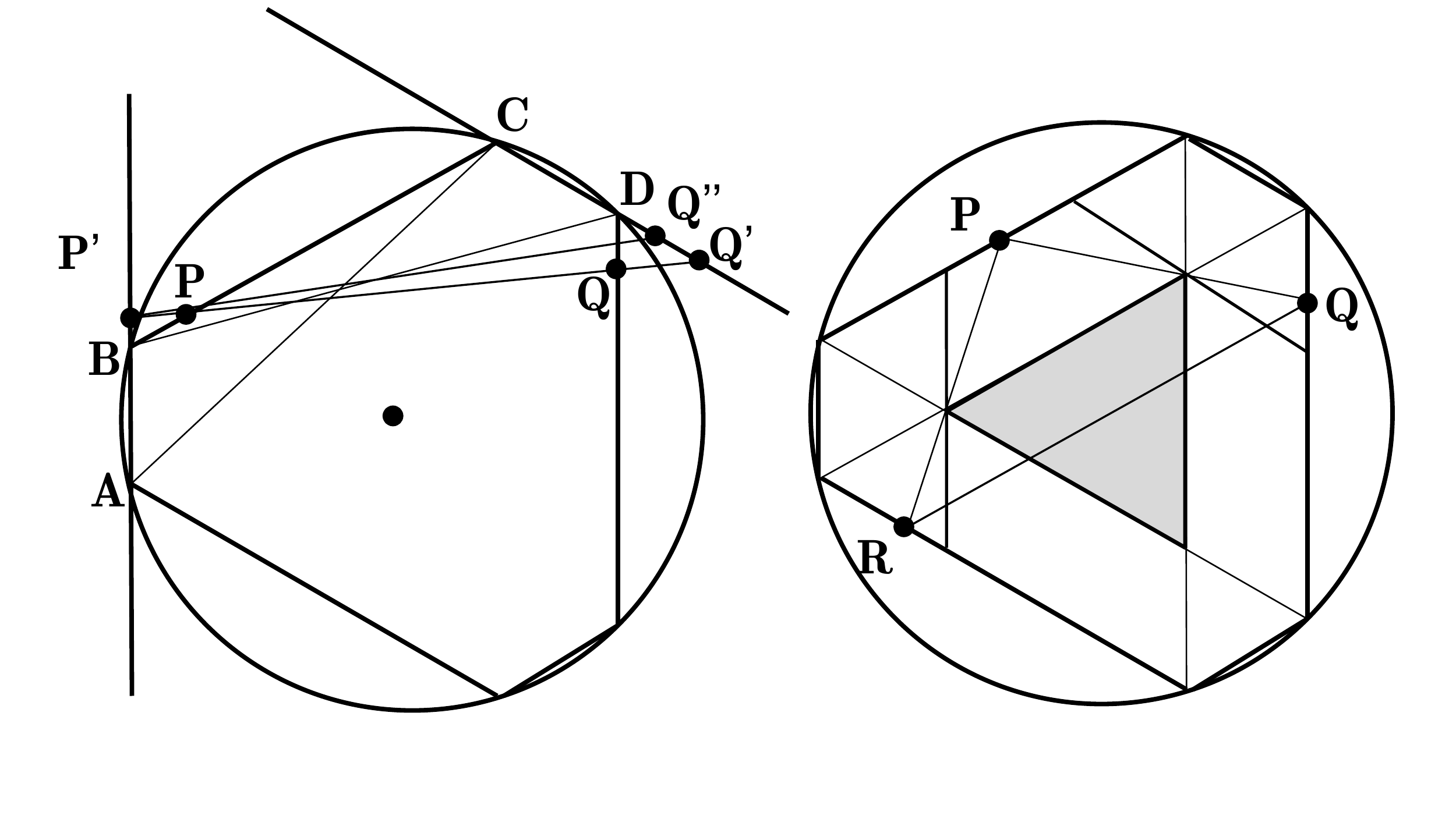}\\
\caption{Proof of Lemma~\ref{lem:gadget}}
  \label{fig:gadgetproof2}
\end{figure}

First observe that if we take $U' = U \cup \{1-s_1, 1-s_2, \ldots, 1-s_N\}$, and generate the set $S = \{I_1, I_2,\ldots, I_{6N}\}$ according to $U'$, then the gadget is further symmetric with respect to flipping along the perpendicular bisector of $BC$. Now without loss of generality $BP \le BC/2$. Since every $I_k$ is now in the intersection of $2N$ triangles, in particular they are also in the intersection of the original $N$ triangles, it suffices to show one of $I_k$ ($k\in [6N]$) is outside halfplane $PQ$.

The first case (left part of Figure~\ref{fig:gadgetproof2}) is when $BP < \epsilon$. In this case we extend $PQ$ to get intersection on $AB$ ($P'$) and intersection on $CD$ ($Q'$). Again since $\angle PMQ \ge 2\pi/3$, we have $DQ\ge BP$. At the same time we know $\angle DQQ'\ge \angle P'PB$, so $DQ' > BP'$. Similar to the previous case, we take $Q''$ so that $CQ'' = AP'$. The intersection of hyperplane $P'Q''$ and the hexagon $ABCDEF$ is at least as large as the intersection of halfplane $PQ$ and the hexagon. When $\epsilon < 1/1000$, we can check $AP' < 2\epsilon \ll 1/50$, therefore we can still view $P'Q''$ as some $A_{2N+1}C_{2N+1}$ for $s_{2N+1} < 2$. Now Lemma~\ref{lem:gadgetpoint} shows there is some vertex $I_k$ not in halfplane $P'Q''$ (and hence not in halfplane $PQ$). 

The final case (right part of Figure~\ref{fig:gadgetproof2}) is when $BP \ge \epsilon$. In this case we notice the triangle with 3 edges $AD$, $BE$, $CF$  (the shaded triangle in the figure) is contained in every $A_iC_iE_i$, thus it must also be in $PQR$. However, since $BC/2\ge BP \ge\epsilon$, we know $AR \le \epsilon$ and $DQ\le\epsilon$. In this case $PQR$ does not even contain the center $M$. 
\end{proof}


\subsection{The Reduction}

Suppose we are given an instance of the $d$-SUM Problem with $N$ values $\{s_1, s_2, ... s_N\}$. We will give a reduction to an instance of Intermediate Simplex in dimension $r - 1 = 3d + 1$. 

To encode the choice of $d$ numbers in the set $\{s_1, s_2, ..., s_N\}$, we use $d$ gadgets defined in Section~\ref{sec:gad}. The final solution of the Intermediate Simplex instance we constructed will include solutions to each gadget. As the solution of a gadget always corresponds to a number in $\{s_1,s_2, ..., s_N\}$ (Lemma~\ref{lem:gadget}) we can decode the solution and get $d$ numbers, and we use an extra dimension $w$ that ``computes'' the sum of these numbers and ensures the sum is equal to $d/2$.

We use three variables $\{x_i,y_i,z_i\}$ for the $i^{th}$ gadget.


\begin{vars}
We will use $3d + 1$ variables: sets $\{x_i, y_i, z_i\}$ for $i \in [d]$ and $w$.
\end{vars}

\begin{cons}[Box]
For all $i \in [d]$, $x_i, y_i \in [0, 1]$, $z_i \in [0,2]$ and also $w \in [0, 1]$.
\end{cons}

\begin{definition}
Let $G \subset \Re^2$ be the hexagon ABCDEF in the two-dimensional gadget 
given in the Section~\ref{sec:gad}.
 Let $H \subset \Re^3$ be the set $conv( \{ (x_i, y_i, z_i) \in \Re^3 | (y_i, z_i) \in G, x_i = 1\}, \vec{0})$. 
\end{definition}

$H$ is a tilted-cone that has a hexagonal base $G$ and has an apex at the origin. 

\begin{definition}
Let $R$ be a $7 \times 3$ matrix and $b \in \Re^7$ so that $\{x | R x \geq b\} = H$. 
\end{definition}

We will use these gadgets to define (some of the) constraints on the polyhedron $P$ in an instance of intermediate simplex:

\begin{cons}[Gadget]
For each $i \in [d]$, $R (x_i, y_i, z_i) \geq b$. 
\end{cons}

Hence when restricted to dimensions $x_i$, $y_i$, $z_i$ the $i^{th}$ gadget $G$ is on the plane $x_i = 1$.

We hope that in a gadget, if we choose three points corresponding to the triangle for some value $s_i$, that of these three points only the point on the $AB$ line will have a non-zero value for $w$ and that this value will be $s_i$. The points on the lines $CD$ or $EF$ will hopefully have a value close to zero. We add constraints to enforce these conditions:

\begin{cons}[CE]
For all $i \in [d]$, $w \leq 1 - y_i + (1- x_i)$
\end{cons}

These constraints make sure that points on $CD$ or $EF$ cannot have large $w$ value.

Recall that we use $z(A)$ to denote the $z$ coordinate of $A$ in the gadget in Section~\ref{sec:gad}. 

\begin{cons}[AB]
For all $i \in [d]$: $w \in \Big [ \frac{ (z_i - z(A)x_i)}{\epsilon} \pm (\frac{10}{\epsilon} y_i  + (1 - x_i)) \Big ]$
\end{cons}

Theses constraints make sure that points on $AB$ have values in $\{s_1, s_2, ..., s_N\}$.

The $AB$ and $CE$ constraints all have the property that when $x_i < 1$ (i.e. the corresponding point is off of the gadget on the plane $x_i = 1$) then these constraints gradually become relaxed. 

To make sure the gadget still works, we don't want the extra constraints on $w$ to rule out some possible values for $x_i$, $y_i$, $z_i$'s. Indeed we show the following claim.

\begin{claim}
For all points in $(x_i, y_i, z_i) \in H$, there is some choice of $w \in [0, 1]$ so that $x_i, y_i, z_i$ and $w$ satisfy the $CE$ and $AB$ Constraints. 
\end{claim}

The proof is by observing that Constraints $AB$ have almost no effect when $y > 0$ and Constraints $CE$ have no effect when $y = 0$.

Constraints 1 to 4 define a polyhedron $P$ in $3d + 1$-dimensional space and furthermore the set of constraints that define $P$ have full rank (in fact even the inequalities in the Box Constraints have full rank). Thus this polyhedron is a valid polyhedron for the Intermediate Simplex problem.

Next we specify the points in $S$ for the Intermediate Simplex problem(each of which will be contained in the polyhedron $P$). Let $I_k$ (for $k \in [3N]$) be the set $S$ in the gadget in Section~\ref{sec:gad}. As before, let $z(I_k)$ and $y(I_k)$ be the $z$ and $y$ coordinates of $I_k$ respectively. 

\begin{definition}[$w$-$\max(I_k)$]
Let $w$-$\max(I_k)$ be the maximum possible $w$-value of any point $I$ with $x_i = 1$, $y_i = y(I_k)$, $z_i = z(I_k)$ and $x_j, y_j, z_j = 0$ for all $j \neq i$ so that $I$ is still contained in $P$.
\end{definition}

\begin{definition}[$O,W, I^i_k, Q$]\label{def:points}
The set $S$ of points for the Intermediate Simplex problem is

\begin{description}

\item{\bf $O$ point:} For all $i \in [d]$, $x_i, y_i, z_i = 0$ and $w = 0$ 

\item{\bf $W$ point:} For all $i \in [d]$, $x_i, y_i, z_i = 0$ and $w = 1$ 

\item{\bf $I^i_k$ points:} For each $i \in [d]$, for each $k \in [3N]$ set $x_i = 1/4$, $y_i = 1/4 y(I_k)$, $z_i = 1/4 z(I_k)$ and for $j \neq i$ set $x_j, y_j, z_j = 0$. Also set $w$ to be the $1/4 \times w$-$\max(I_k)$. 

\item{\bf $Q$ point:} For each $i \in [d]$, $x_i = 1/d$, $y_i = y(M)/d$, $z_i = z(M)/d$ and $w = 1/6$ 

\end{description}

\end{definition}

This completes the reduction of $3$-SUM to intermediate simplex, and next we establish the COMPLETENESS and SOUNDNESS of this reduction. 

\subsection{Completeness and Soundness}

The completeness part is straight forward: for $i^{th}$ gadget we just select the triangle that corresponds to $s_{k_i}$.

\begin{lemma}
If there is a set $\{s_{k_1}, s_{k_2}, ... s_{k_d}\}$ of $d$ values (not necessarily distinct) such that $\sum_{i \in [d]} s_{k_i} = d/2$ then there is a solution to the corresponding Intermediate Simplex Problem. 
\end{lemma}

\begin{proof}
We will choose a set of $3d + 2$ points $T$: We will include the $O$ and $W$ points, and for each $s_{k_i}$, we will choose the triangle corresponding to the value $s_{k_i}$ in the $i^{th}$ gadget. Recall the triangle is $A_{k_i}C_{k_i}E_{k_i}$ in the gadget defined in Section~\ref{sec:gad}. The points we choose have $x_i = 1$ and $y_i$, $z_i$ equal to the corresponding point in the gadget. We will set $w$ to be $s_{k_i}$ for the point on the line $AB$ and we will set $w$ to be zero for the other two points not contained in the line $AB$. The rest of the dimensions are all set to 0.

Next we prove that the convex hull of this set of points $T$ contains all the points in $S$: The points $O$ and $W$ are clearly contained in the convex hull of $T$ (and are in fact in $T$!). Next consider some point $I^i_k$ in $S$ corresponding to some intersection point $I_k$ in the gadget $G$. Since $I_k$ is in the convex hull of the triangle corresponding to $s_{k_i}$ in the gadget $G$, there is a convex combination of the these three points $A_{k_i}, C_{k_i}, E_{k_i}$ in $T$ (which we call $J$) so that $1/4 J$ matches $I^i_k$ on all coordinates except possibly the $w$-coordinate. Furthermore the point $J$ has some value in the coordinate corresponding to $w$ and this must be at most the corresponding value in $I^i_k$ (because we chose the $w$-value in $I^i_k$ to be $1/4 \times w$-$\max(I_k)$). Hence we can distribute the remaining $3/4$ weight among the $O$ and $W$ points to recover $I^i_k$ exactly on all coordinates. 

Lastly, we observe that if we equally weight all points in $T$ (except $O$ and $W$) we recover the point $Q$. In particular, the $w$ coordinate of $Q$ should be $\frac{1}{3d}\sum_{i=1}^d s_{k_i} = 1/6$. 
\end{proof}

Next we prove SOUNDNESS for our reduction. Suppose the solution is $T$, which is a set of $3d+2$ points in the polyhedron $P$ and the convex hull of points in $T$ contains all the $O$, $W$, $I^i_k$, $Q$ points (in Definition~\ref{def:points}).

\begin{claim}
The points $O$ and $W$ must be in the set $T$.
\end{claim}

\begin{proof}
The points $O$ and $W$ are vertices of the polyhedron $P$ and hence cannot be expressed as a convex combination of any other set of points in $P$. 
\end{proof}

Now we want to prove the rest of the $3d$ points in set $T$ is partitioned into $d$ triples, each triple belongs to one gadget. Set $T' = T - \{O\} - \{W\}$. 

\begin{definition}
For $i \in [d]$, let $$T'_i = \{Z \in T' | j \neq i \Rightarrow x_j(Z), y_j(Z), z_j(Z) = 0 \mbox{ and one of } x_i(Z), y_i(Z), z_i(Z) \neq 0\}$$ 
\end{definition}

\begin{claim}
The sets $T'_i$ partition $T'$ and each contain exactly $3$ nodes. 
\end{claim}

\begin{proof}
The sets $T'_i$ are disjoint, and additionally each set $T'_i$ must contain at least $3$ nodes (otherwise the convex hull of $T'_i$ even restricted to $x_i, y_i, z_i$ cannot contain the points $I^i_k$). This implies the Claim. 
\end{proof}

Recall the gadget in Section~\ref{sec:gad} is a two dimensional object, but it is represented as a three dimensional cone in our construction. We would like to apply Lemma~\ref{lem:gadget} to points on the plane $x_i = 1$ (in this plane the coordinates $y_i$,$z_i$ act the same as $y$, $z$ in the gadget). 

\begin{definition}
For each point $Z \in T'_i$, let $ext(Z) \in \Re^3$ be the intersection of the line connecting the origin and $(x_i(Z), y_i(Z), z_i(Z))$ with the $x_i = 1$ base of the set $\{ (x_i, y_i, z_i) | R (x_i, y_i, z_i) \geq b\}$. Let $ext(T'_i)$ be the point-wise $ext$ operation applied to each point in $T'_i$. 
\end{definition}

Since the points $I^i_k$ are in the affine hull of $T'_i$ when restricted to $x_i,y_i,z_i$ , we know $ext(I^i_k)$ must be in the convex hull of $ext(T'_i)$. Using Lemma~\ref{lem:gadget} in Section~\ref{sec:gad}, we get:

\begin{corollary}
$ext(T'_i)$ must correspond to some triangle $A_{k_i}C_{k_i}E_{k_i}$ for some value $s_{k_i}$. 
\end{corollary}

Now we know how to decode the solution $T$ and get the numbers $s_{k_i}$. We will abuse notation and call the 3 points in $T'_i$ $A_{k_i}$, $C_{k_i}$, $E_{k_i}$ (they were used to denote the corresponding points in the 2-d gadget in Section~\ref{sec:gad}).We still want to make sure the $w$ coordinate correctly ``computes'' the sum of these numbers. As a first step we want to show that the $x_i$ of all points in $T'_i$ must be 1 (we need this because the Constraints AB and CE are only strict when $x_i = 1$).

\begin{lemma}\label{lemma:setx}
For each point $Z \in T'_i$, $x_i(Z) = 1$
\end{lemma}

\begin{proof}
Suppose, for the sake of contradiction, that $x_i(Z) < 1$ (for $Z \in T'_i$). Then consider the point $Q$. Since $\sum_{i \in [d]} x_i(Q) = 1$, and for any point in $T$ $\sum_{i \in [d]} x_i \le 1$, there is no convex combination of points in $T$ that places non-zero weight on $Z$ and equals $Q$. 

Let $T''_i$ be $T'_i\backslash\{Z\}$, we observe that the points in $T''_i$ are the only points in $T$ that have any contribution to $(x_i,y_i,z_i)$ when we want to represent $Q$ (using a convex combination). For now we restrict our attention to these three dimensions.
 When trying to represent $Q$ we must have $1/d$ weight in the set $T''_i$ (because of the contribution in $x_i$ coordinate). The $y_i$, $z_i$ coordinates of $Q$ are $y(M)/d$, $z(M)/d$ respectively. This means if we take projection to $y_i,z_i$ plane $M$ must be in the convex hull of $T''_i$. However that is impossible because no two points in $A_kC_kE_k$  contain $M$ in their convex hull. This contradiction implies the Lemma. 
\end{proof}

\begin{lemma}
\label{lem:equalweight}
Any convex combination of points in $T$ that equals the point $Q$ must place equal weight on all points in $T'$. 
\end{lemma}

\begin{proof}
Using Lemma~\ref{lemma:setx}, we conclude that the total weight on points in $T'_i$ is exactly $1/d$, and there is a unique convex combination of the points $T'_i$ (restricted to $y_i, z_i$) that recover the point $M$ which is the $1/3, 1/3, 1/3$ combination. This implies the Lemma. 
\end{proof}

Now we are ready to compute the $w$ value of the point $Q$ and show the sum of $s_{k_i}$ is indeed $d/2$.

\begin{lemma}[Soundness]
When $\epsilon < N^{-Cd}$ for some large enough constant $C$, if there is a solution to the Intermediate Simplex instance, then there is a choice of $d$ values that sum up to exactly $d/2$. 
\end{lemma}

\begin{proof}
As we showed in previous Lemmas, the solution to the Intermediate Simplex problem must contain $O$, $W$, and for each gadget $i$ the solution has 3 points $T'_i$ that correspond to one of the solutions of the gadget. Suppose for gadget $i$ the triangle we choose is $A_{k_i}C_{k_i}E_{k_i}$. By Constraints $AB$ we know $w(A_{k_i}) = s_{k_i}$, by Constraints $CE$ we know $w(C_{k_i}) \le \epsilon$ and $w(E_{k_i}) \le \epsilon$.

By Lemma~\ref{lem:equalweight} there is only one way to represent $Q$, and $w(Q) = \frac{1}{3d} \sum_{i=1}^d[ w(A_{k_i})+w(C_{k_i})+w(E_{k_i})] = 1/6$. 

\begin{equation}
\sum_{i=1}^d s_{k_i} = \sum_{i=1}^d w(A_{k_i}) = \frac{d}{2} - \sum_{i=1}^d [w(C_{k_i}) + w(E_{k_i})].
\end{equation}

Since $w(C_{k_i})$ and $w(E_{k_i})$'s are small, we have $\sum_{i=1}^d s_{k_i} \in [d/2 - 2d\epsilon, d/2]$. However the numbers only have $O(d\log N)$ bits and $\epsilon$ is so small, the only valid value in the range is $d/2$. Hence the sum $\sum_{i=1}^d s_{k_i}$ must be equal to $d/2$. 
\end{proof}



\section{Fully-Efficient Factorization under Separability}\label{sec:sep}

Earlier, we gave algorithms for NMF, and presented evidence that no $(nm)^{o(r)}$ time algorithm exists for determining if a matrix $M$ has nonnegative rank at most $r$. Here we consider conditions on the input that allow the factorization to be found in time polynomial in $n$, $m$ and $r$. (In Section~\ref{subsec:separablenoise}, we give a noise-tolerant version of this algorithm). To the best of our knowledge this is the first example of an algorithm (that runs in time poly$(n,m,r)$) and provably works under a non-trivial condition on the input. Donoho and Stodden~\cite{DS} in a widely-cited paper identified sufficient conditions for the factorization to be unique (motivated by applications of NMF to a database of images) but gave no algorithm for this task.
We give an algorithm that runs in time poly$(n,m,r)$ and assumes only one of their conditions is met (separability). We note that this separability condition is quite natural in its own right, since it is usually satisfied \cite{Blei} by model parameters fitted to various generative models (e.g. LDA \cite{LDA} in information retrieval).



\begin{definition}[Separability]
A nonnegative factorization $M = AW$ is called {\em separable} if for each $i$  there is some row $ f(i)$ of $A$ that has a single nonzero entry and this entry is in the $i^{th}$ column. 
\end{definition}

Let us understand this condition at an intuitive level in  context of clustering documents by topic, which was discussed in the introduction. Recall that there a column of $M$ corresponds to a document. Each column of $A$  represents a topic and its entries specify the probability that a  word occurs in that  topic.  
 The NMF thus ``explains'' the $i^{th}$ document as 
$A W_i$ where the column vector $W_i$ has (nonnegative) coordinates summing to one---in other words,  $W_i$ represents a convex combination of topics. In practice, the total number of words $n$ may number in the thousands or tens of thousands, and the number of topics in the dozens. Thus it is not unusual to find factorizations in which each topic is flagged by a word that appears only in that topic and not in the other topics~\cite{Blei}. The separability condition asserts that this happens for every topic\footnote{More realistically, the word may appear in other topics only with negligible property instead of zero probability. This is allowed in our noise-tolerant algorithm later.}.

For simplicity we assume without loss of generality that the rows of $M$ are normalized to have unit $\ell_1$-norm. After normalizing $M$, we can still normalize $W$ (while preserving the factorization) by re-writing the factorization as $M = AW = (AD) (D^{-1} W)$ for some $r \times r$ nonnegative matrix $D$. By setting $D_{i,i} = \norm{W^i}_1$ the rows of $D^{-1} W$ will all have $l_1$ norm 1. When rows of $M$ and $W$ are all normalized the rows of $A$ must also have unit $\ell_1$-norm because 

$$1 = \norm{M^i}_1 = \norm{\sum_{j=1}^r A_{i,j} W^j}_1 = \sum_{j=1}^r A_{i,j} \norm{W^j}_1 = \sum_{j=1}^r A_{i,j}.$$

The third equality uses the nonnegativity of $W$. Notice that after this normalization, if a row of $A$ has a unique nonzero entry (the rows in Separability), that particular entry must be one.

We also assume $W$ is a simplicial matrix defined as below.


\begin{definition}[simplicial matrix]
A nonnegative matrix $W$ is {\em simplicial} if no row in $W$ can be represented in the convex hull of the remaining rows in $W$. 
\end{definition}
The next lemma shows that without loss of generality we may assume $W$ is simplicial.

\begin{lemma}\label{lemma:alt}
If a nonnegative matrix $M$ has a separable factorization $AW$ of inner-dimension at most $r$ then there is one in which $W$ is simplicial. 
\end{lemma}
\begin{proof}
Suppose $W$ is not simplicial, and let the $j^{th}$ row $W^j$  be in the convex hull of the remaining rows. Then we can represent $W^j = \vec{u}^T W$ where $\vec{u}$ is a nonnegative vector with $|\vec{u}|_1 = 1$ and the $j^{th}$ coordinate is 0.

Now modify $A$ as follows. For each row $A^{j'}$ in $A$ that has a non-zero $j^{th}$ coordinate, we zero out the $j^{th}$ coordinate and add $A^{j'}_j \vec{u}$ to the row $A^{j'}$. At the end the matrix  is still nonnegative but whose $j^{th}$ column is all zeros. So delete the $j^{th}$ column and let the resulting $n \times (r-1)$ matrix be $A'$. Let $W'$ be the matrix obtained by deleting the $j^{th}$ row of $W$. Then by construction we have $M = A' W'$.  Now we claim $A'$ is separable.

Since $A$ was originally separable, for each column index $i$  there is some row, say the $f(i)^{th}$ row, that  has a non-zero entry in the $i^{th}$ column and zeros everywhere else. 
If $i\neq j$ then by definition the above operation does not change  the $f(i)^{th}$ row of $A$. If $i=j$  the $j^{th}$ index is deleted at the end. In either case the final matrix $A'$ satisfies the separability condition.

Repeating the above operation for all violations of the simplicial condition we end with a separable factorization of $M$ (again with inner-dimension at most $r$) where $W$ is simplicial. 
\end{proof}

\begin{theorem}
There is an algorithm that  runs in time polynomial in $n$, $m$ and $r$   and given a  matrix $M$ outputs a separable factorization with inner-dimension at most $r$
(if  one exists).
\end{theorem}

\begin{proof}
We can apply Lemma~\ref{lemma:alt} and assume without loss of generality that there is a factorization $M = AW$ where $A$ is separable and $W$ is simplicial. The separability condition implies that every row of $W$ appears among the rows of $M$. Thus $W$ is hiding in plain sight in $M$; we now show how to find it. 

Say a row $M^j$ is a {\em loner} if (ignoring other rows that are copies of $M^j$) it is not in the convex hull of the remaining rows.
The simplicial condition implies that the rows of $M$ that correspond to rows of $W$ are loners.

\begin{claim}
A row $M^j$ is a loner iff $M^j$ is equal to some row $W^{i}$
\end{claim}

\begin{proof}
Suppose (for contradiction) that a row in $M^j$ is not a loner and but it is equal to some row $W^i$. Then there is a set  $S$ of rows of $M$ so that $M^j$ is in their convex hull and furthermore for all $j' \in S$, $M^{j'}$ is not equal to $M^j$. Thus there is a nonnegative vector $u \in \Re^n$ that is 0 at the $j^{th}$ coordinate and positive on indices in $S$
such that $u^T M =M^j$. 

Hence $u^T A W = M^j = W^i$, but $u^T A$ must have unit $\ell_1$-norm (because $\norm{u}_1 = 1$, all rows of $A$ have unit $\ell_1$-norm and are all nonnegative), also $u^T A$ is non-zero at position $j'$. Consequently $W^i$ is in the convex hull of the other rows of $W$, which yields a contradiction.

Conversely if a row $M^j$ is not equal to any row in $W$, we conclude that $M^j$ is in the convex hull of the rows of $W$. Each row of $W$ appears as a row of $A$ (due to the separability condition). Hence $M^j$ is not a loner  because $M^j$ is in the convex hull of rows of $M$ that are equivalent to $M^j$ itself. 
\end{proof}

Using linear programming, we can determine which rows $M^j$ are loners. Due to separability there will be exactly $r$ different loner rows, each corresponds to one of the $W^i$. Thus we are able to recover $W'$ that is equal to $W$ after permutation over rows.
We can compute a nonnegative $A'$ such that $A' W' = M$, and such solution $A'$ is necessarily separable (since it is just  equal to $A$ after permutation over columns). 
\end{proof}





\subsection{Adding Noise}
\label{subsec:separablenoise}
In any practical setting the data matrix $M$ will not have an exact NMF of low inner dimension since its entries are invariably subject to noise. Here we consider how to extend our separability-based algorithm to work in presence of noise.
We assume that the input matrix  $M'$ is obtained by perturbing each row of $M$ by adding a vector of $\ell_1$-norm at most $\epsilon$, where 
$M$ has a separable factorization of inner-dimension $r$. Alternatively, $\norm{M'^i - M^i}_1 \le \epsilon$ for all $i$. Notice that the case in which the separability condition is only approximately satisfied is a subcase of this: If for each column there is some row in which that column's entry  is at least $1-\epsilon$ and the sum of the other row entries is less than $\epsilon$ then the matrix $M'$ will satisfy the condition stated above. 
(Note that $M, A, W$ have been scaled as discussed above.)

Our algorithm will require one more condition -- namely, we require the
unknown matrix $W$ to be ``robustly'' simplicial instead of just simplicial. 



\begin{definition}[$\alpha$-robust simplicial] We call $W$ {\em $\alpha$-robust simplicial} if no row in $W$ has $\ell_1$ distance smaller than $\alpha$ to the convex hull of the remaining rows in $W$. (Here all rows have unit $\ell_1$-norm.)
\end{definition}

Recall from Lemma~\ref{lemma:alt} that the simplicial condition can be assumed without loss of generality under separability. In general $\alpha$-robust simplicial condition does not follow from separability. However, any reasonable generative model would surely posit that the matrix $W$ ---whose columns after all represents distributions---satisfies the condition above. For instance, if columns of $W$ are picked randomly from the unit $\ell_1$ ball then after normalization $\alpha$ is more than $1/10$. Regardless of whether or not one self-identifies as a bayesian, it seems reasonable that any suitably generic way of picking column vectors would tend to satisfy the $\alpha$-robust-simplicial property.

\begin{theorem}
\label{thm:separablenoise}
Suppose $M=AW$ where 
$A$ is separable and $W$ is $\alpha$-robust simplicial. Let $\epsilon$ satisfy 
$20\epsilon/\alpha+13\epsilon < \alpha$. Then  there is a polynomial time algorithm that given $M'$ such that for all rows $\norm{M'^i - M^i}_1<\epsilon$, finds a nonnegative matrix factorization
$A'W'$ of the same inner dimension such that the $\ell_1$ norm of each row of $M' -A'W' $ is at most $10\epsilon/\alpha + 7\epsilon$.
\end{theorem}

\begin{proof}
Separability implies that for any column index $i$ there is a  row $f(i)$ in $A$ whose only nonzero entry is in the $i^{th}$ column.
Then $M^{f(i)} = W^i$ and consequently $\norm{M'^{f(i)} - W^i}_1 < \epsilon$. Let us call these rows $M'^{f(i)}$ for all $i$ the {\em canonical rows}. 
From the above description the following claim is clear since 
the rows of $M$ can be expressed as a convex combination of $W^i$'s. 

\begin{claim}
\label{claim:repsentationwitherror}
Every row $M'^j$ has $\ell_1$-distance at most $2\epsilon$ to the convex hull of canonical rows.
\end{claim}
\begin{proof}
$$\norm{M'^j - \sum_{k=1}^r A_{j,k} M'^{f(k)}}_1 \le \norm{M'^j-M^j}_1 + \norm{M^j - \sum_{k=1}^r A_{j,k} M^{f(k)}}_1+\norm{\sum_{k=1}^r A_{j,k} (M^{f(k)} - M'^{f(k)})}_1$$
and we can bound the right hand side by $2 \epsilon$.
\end{proof} 

Next, we show how to find the canonical rows.  For a row $M'^j$, we call it a {\em robust-loner}
 if upon ignoring rows whose $\ell_1$ distance to $M'^j$ is less than $d = 5\epsilon/\alpha+2\epsilon$,  the $\ell_1$-distance
of $M'^j$  to the convex hull of the remaining rows  is more than $2\epsilon$. Note that we can identify robust-loner rows using
linear programming.

The following two claims establish that a row  of $M'^j$ is a  robust-loner if and only if it is close to some row $W^i$.

\begin{claim}
If $M'^j$ has distance more than $d+\epsilon$ to all of the $W^i$'s, then it cannot be a robust loner.
\end{claim}
\begin{proof}
Such an $M'^j$ has distance at least $d$ to each of the canonical rows.   The previous claim shows $M'^j$ is close to the convex hull of the canonical rows and thus by definition it cannot be a robust-loner.
\end{proof}

\begin{claim}
All canonical rows are robust-loners.
\end{claim}

\begin{proof}
Since $\norm{M'^{f(i)}-W^i}_1 \le \epsilon$, when we check if $M'^{f(i)}$ is a robust-loner (using linear programming), we leave out of consideration all rows that have $\ell_1$-distance at most $5\epsilon/\alpha + \epsilon$ to $W^i$. In particular, this omits any  row $M'^j$ such that $M^j = \sum_{k=1}^r A_{j,k} W^k$ and $A_{j,i} \ge 1 - 5\epsilon/\alpha$. All remaining rows have $A_{j,i} \le 1-5\epsilon/\alpha$, and hence the $\ell_1$ distance of $W^i$ to 
$conv(W\backslash W^i)$ is at least  $ \alpha$ (by the $\alpha$-robust simplicial property), we conclude that the distance between $W^i$ and the convex hull of remaining $M^j$'s must be at least $5\epsilon/\alpha * \alpha = 5\epsilon$. Since $M'$ is close to $M$ the $\ell_1$-distance between $M'^{f(i)}$ and the convex hull of remaining rows $M'^j$'s must be at least $5\epsilon-2\epsilon = 3\epsilon$. Therefore $M'^{f(i)}$ is a robust-loner.
\end{proof}

The previous claim implies that each robust-loner row is within $\ell_1$-distance $d+\epsilon$ to some $W^i$ 
and conversely, for every $W^i$ there is at least one robust-loner row that is close to it. Since the $\ell_1$-distances between $W^i$'s are at least $4(d+\epsilon)$, we can apply distance based clustering
on the robust-loner rows: place two robust-loner rows into the same cluster if and only if these rows are within $\ell_1$-distance at most $2(d+\epsilon)$. Clearly we will obtain $r$ clusters, one corresponding to each of the $W^i$'s. Choose one row from each of the cluster, and using similar argument as Claim~\ref{claim:repsentationwitherror} we deduce that every row of $M'$ is within $2(d+\epsilon)+\epsilon = 10\epsilon/\alpha+7\epsilon$ to the convex hull of the rows we selected. Therefore these rows form a nonnegative $W'$ and we can find $A'$ so that $\norm{M'^j- (A'W')^j}_1 \le 10\epsilon/\alpha+7\epsilon$ for all $j$.
\end{proof}

\section{Approximate Nonnegative Matrix Factorization}\label{sec:adv}

Here we consider the case in which the given
matrix does not have an exact low-rank NMF but rather
can be approximated by a nonnegative factorization with small inner-dimension. 
We refer to this as {\em Approximate} NMF. Unlike the algorithm in Theorem~\ref{thm:separablenoise}, the algorithm here works with general nonnegative matrix factorization: we do not make any assumptions on matrices $A$ and $W$. Throughout this section we will use $\norm{}_F$ to denote the Froebenius norm, $\norm{}_2$ to
denote the spectral norm and $\norm{}$ applied to a vector will denote the standard Euclidean norm. 

\begin{theorem}
Let $M$ be an $n\times m$ nonnegative matrix
such that there is a factorization $AW$ satisfying
 $\norm{M-AW}_F \le \epsilon \norm{M}_F$, where $A$ and $W$ are
 nonnegative and have inner-dimension $r$. There
 is an algorithm that computes $A'$ and $W'$ satisfying 
 $$\norm{M-A'W'}_F \le O(\epsilon^{1/2}r^{1/4})\norm{M}_F$$ in
  time $2^{\poly(r\log(1/\epsilon))}\poly(n,m)$.
\end{theorem}

Note that the matrix $M$ need not have low rank, but we will be able to assume $M$ has rank at 
most $r$ without loss of generality: Let $M'$ be the best rank at most $r$ approximation (in terms of Frobenius norm) to $M$. 
This can be computed using a truncated singular value decomposition (see e.g. \cite{GV}). Since $A$ and $W$ have inner-dimension $r$, we get:

\begin{claim}
$\norm{M' - M}_F \leq \norm{M-AW}_F$
\end{claim}

Throughout this section, we will assume that the input matrix $M$ has rank at most $r$ - since otherwise
we can compute $M'$ and solve the problem for $M'$. Then using the triangle inequality, 
any good approximation to $M'$ will also be a good approximation to $M$.


Throughout this section, we will use the notation $A_t$ to denote the $t^{th}$ column of $A$
and $W^t$ to denote the $t^{th}$ row of $W$. 
Note that $W^t$ is a row vector so we will frequently use $A_t W^t$ to denote an outer-product. 
Next, we apply a simple re-normalization that will
allow us to state the main steps in our algorithm in a more friendly notation. 


\begin{lemma}\label{lemma:normalization}
We can assume without loss of generality that for all $t$
\begin{align}
\norm{W^t}& = 1\; \label{|Yt|} \\
\norm{A_t}&\leq  (1 + \epsilon) \norm{M}_F\; \label{|Xt|}
\end{align}
and furthermore $\norm{A}_F \le (1 + \epsilon) \norm{M}_F$.
\end{lemma}

\begin{proof}
We can write $AW= \sum_{t=1}^r A_t W^t.$ So we may scale $A_t,W^t$ to ensure that
$\norm{W^t}=1$. Next, since $A$ and $W$ are nonnegative we have
$\norm{AW}_F\geq \norm{A_t W^t}_F = \norm{A_t}\norm{W^t}$ and $\norm{AW}_F\leq  (1 + \epsilon) \norm{M}_F$ and this implies the first condition in the lemma. 

Next we observe
\[
\norm{AW}_F^2 = \sum_{i=1}^n \sum_{j=1}^m \Big [\sum_{t=1}^r (A_t W^t)_{i,j} \Big ]^2 \ge \sum_{t=1}^r\sum_{i=1}^n \sum_{j=1}^m [A_t W^t]_{i,j}^2 = \norm{A}_F^2.
\]
where the inequality follows because all entries in $A$ and $W$ are nonnegative, and the last equality follows because $\norm{W^t} = 1$. 
\end{proof}

Note that this lemma immediately implies that $\norm{W}_F \leq \sqrt{r}$. 

The intuition behind our algorithm  is to decompose the unknown matrix $W$ as 
the sum of two parts: $W= W_0+W_1$. The first part $W_0$ is responsible for 
 how good $AW$ is as an approximation to $M$ (i.e., $\norm{M-AW_0}_F$ is small) but could be
negative; the second part $W_1$ has little effect on the approximation
 but is important in ensuring the sum $W_0+W_1$ is nonnegative. 
The algorithm will find good approximations to $W_0, W_1$.

What are $W_0, W_1$? Since removing $W_1$ has little effect on
how good $AW$ is as an approximation to $M$, this matrix should be roughly the projection of $W$ onto the ``less significant'' 
singular vectors of $A$. Namely, let the singular value decomposition of $A$ be
\begin{align} A & = \sum_{t=1}^r \sigma_t u_t
v_t^T. \label{eqn:svdA}
\end{align}
and suppose that $\sigma_1 \geq \sigma_2 .... \geq \sigma_r$. Let $t_0$ be the largest $t$ for which $|\sigma_t|\ge
\delta \norm{M}_F$ (where $\delta$ is a constant that is polynomially
related to $r$ and $\epsilon$ and will be specified later). Then set
\begin{align}W_0 &=
\sum_{t=1}^{t_0} (v_t v_t^T) W; \qquad 
W_1 = \sum_{t=t_0+1}^r (v_tv_t^T) W.
\end{align} 

\begin{lemma} \label{lemma:w0}
$\norm{M - AW_0}_F \le \epsilon \norm{M}_F + \delta \sqrt{r} \norm{M}_F$
\end{lemma}
\begin{proof}
By the triangle inequality $\norm{M - AW_0}_F \le \norm{M - AW}_F +
\norm{AW_1}_F$. Also $AW_1 = \sum_{t = t_0+1}^r \sigma_t (u_t
v_t^T) W$, so we have

\[
\norm{AW_1}_F = \norm{\sum_{t = t_0+1}^r \sigma_t (u_t v_t^T) W}_F \le  \norm{\sum_{t = t_0+1}^r \sigma_t (u_t v_t^T)}_2 \norm{W}_F \le \delta  \norm{M}_F \sqrt{r},
\]
where the last inequality follows because $\norm{W}_F \leq \sqrt{r}$ and the spectral norm of $\sum_{t = t_0+1}^r (u_t v_t^T)$ is one. 
\end{proof}

Next, we establish a lemma that will be useful when searching for (an approximation to) $W_0$:

\begin{lemma} \label{lem:w0pp}
There is an $r \times m$ matrix $W_0'$  such that each row is in the
span of the rows of $M$ and which satisfies $\norm{W_0' -W_0}_F \leq
2\epsilon / \delta.$
\end{lemma}
\begin{proof}
Consider the matrix  $A^+ =\sum_{t=1}^{t_0} \frac{1}{\sigma_t}v_t
u_t^T$. Thus $A^+$ is a pseudo-inverse of the truncated SVD of $A$.
Note that $W_0 = A^+AW$ and the spectral norm $\norm{A^+}_2$
is at most $1/(\delta\norm{M}_F)$.
Then we can choose $W_0' = A^+ M$. Clearly, each row of $W_0'$ is
in the span of the rows of $M$. Furthermore, we have
\[
\norm{W_0' - W_0}_F = \norm{A^+(M - AW)}_F \le \norm{A^+}_2 \norm{M-AW}_F \le \frac{1}{\delta\norm{M}_F}\cdot 2\epsilon \norm{M}_F \le \frac{2\epsilon }{\delta}.
\]
\end{proof}

\begin{lemma}
\label{lemma:w1}
There is an algorithm that in time
$2^{\poly(r\log(1/\epsilon))}\poly(n,m)$  finds $W_0'', W_1'$ and $A'$ such that $W_0''+W_1' \ge 0$, $A'\ge 0$ and

\[
\norm{M - A'(W_0''+W_1')}_F \le O(\frac{\epsilon}{\delta} \norm{A}_F + \epsilon \norm{M}_F + \delta \sqrt{r} \norm{M}_F).
\]
\end{lemma}
\begin{proof}
We use exhaustive enumeration to find a close
approximation to the matrix  $W_0'$ of Lemma~\ref{lem:w0pp}, and then we
use convex programming to find $W_1', A'$:

The exhaustive enumeration is simple: try all vectors that lie in some
 $\epsilon_1$-net in the span of the rows of $M$, where $\epsilon_1 =
 \epsilon/\delta$ .  Such an $\epsilon_1$-net is easily enumerated in
 the provided time since the row vectors are  smaller than $\norm{W}_F
 = \sqrt{r}$ and their span is $r$-dimensional.
Contained in this net there must an $W_0''$ such that $\norm{A^+M -
  W_0''}_F \le \epsilon_1$. Using Lemma~\ref{lem:w0pp},
$\norm{W_0-W_0'}_2 \le 2\epsilon /\delta,$ so the triangle
inequality implies $\norm{W_0 -W_0''}_F \leq 2\epsilon / \delta
+ \epsilon_1   \le 4\epsilon /\delta$. 

Next, we give a method to find 
suitable substitutes $W_1', A'$  for $W_1, A$ respectively so 
that $W_0'+W_1' \ge 0$ and $A'(W_0'+W_1')$ is a good approximation to $M$.

Let us assume we know the vectors $v_i$ appearing in the SVD
expression~(\ref{eqn:svdA})  and $\norm{A}_F$. This is easy to guarantee since
we can enumerate over all choices of the $v_i$'s (which are unit vectors in
$\Re^r$) using a suitable $\epsilon_2$-net where
$\epsilon_2 = \min\{ \frac{\epsilon}{\delta r}, 0.1\}$. Also, $\norm{A}_F$ is a scalar value that can be easily guessed within multiplicative factor $1.01$.

Let $W_1' = Z$ be
the optimal solution to the following convex program:

\begin{align}
\min \quad &\norm{A}_F^2 \sum_{t=1}^{t_0} \norm{v_i^T Z}^2 + \delta^2
\norm{M}_F^2 \sum_{t=t_0+1}^r \norm{v_i^T Z}^2 \label{eqn:convexopt}\\
s.t. \quad &W_0''+Z \ge 0.
\end{align}

This is optimization problem is convex since the constraints are linear and the
objective function is quadratic but convex. (In fact this optimization problem can be separated into $m$ smaller convex programs because the constraints between different columns of $W_1'$ are independent).

When the vectors we enumerated (denoted as $\{v_i'\}$) are close enough to the true values $\{v_i\}$, that is, when $\sum_{i=1}^r \norm{v_i'-v_i}^2 \le \min\{\frac{\epsilon^2}{\delta^2 r}, 0.01\}$, the value of the objective function after substituting $v$ by $v'$ can only change by at most $O(\frac{\epsilon^2}{\delta^2} \norm{A}_F^2 + r \delta^2 \norm{M}_F^2)$.  From now on we work with the true values of $\{v_i\}$. The Claim below and arguments after will still be true although the vectors are not exact.

\begin{claim}
The optimal value of this convex program is at most $O(\frac{\epsilon^2}{\delta^2} \norm{A}_F^2 + r \delta^2 \norm{M}_F^2)$.
\end{claim}
\begin{proof} We prove that $W_1' = W-W_0'' =
  (W_0-W_0'')+W_1$ is a feasible solution and that the objective value of this solution is the value claimed in the lemma. 

 Since $W_1 = \sum_{t=t_0+1}^r (v_t v_t^T) W$ only
  contributes to the second term of the objective function
in (\ref{eqn:convexopt}), we can upper bound the objective as
 $$\norm{W_0-W_0''}_F^2 \norm{A}_F^2 +
 (\norm{W_0-W_0''}_F+\norm{W_1}_F)^2 \delta^2\norm{M}_F^2.$$
 The proof is completed because $\norm{W_0-W_0''}_F = O(\frac{\epsilon}{\delta})$ and $\norm{W_1}_F \le \norm{W}_F = \sqrt{r}$. 
\end{proof}

After solving the convex program, we obtaine a candidate 
$W_1'$. Let $W' = W_0''+W_1'$. To get the right $A'$ (since $W'$ is fixed) we can find the $A'$ that minimizes
$\norm{M-A'W'}_F^2$ by solving a least-squares problem. Clearly such an $A'$ satisfies
$\norm{M - A'(W_0''+W_1')}_F\le \norm{M - A(W_0''+W_1')}_F $
and the latter quantity is bounded  by
$\norm{M - AW_0}_F + \norm{A(W_0-W_0'')}_F + \norm{AW_1'}_F.$

Lemma~\ref{lemma:w0} bounds the first term and Lemma~\ref{lem:w0pp} bounds the second term. The
square of the last term is bounded by the objective function of
the convex program. 
\end{proof}

Finally, by choosing $\delta =\frac{ \sqrt{\epsilon}}{r^{1/4}}$ we get $A'$, $W' = W_0''+W_1'$ such that $\norm{M-A'W'}_F \le O(\epsilon^{1/2} r^{1/4}) \norm{M}_F$.

\section*{Concluding Remarks}

Here, we initiated a rigorous study of nonnegative matrix factorization. 
Our hardness result rules out significant improvements over our worst-case
results for fixed inner-dimension $r$. We believe that our $\mbox{poly}(m,n,r)$-time algorithm for finding separable factorizations may point the way for future work. What other plausible conditions can one impose on the factors in real-life applications? We also hope our work promotes further theoretical study of nonnegative rank. 

This work is part of a broader agenda of bringing greater rigor to the analysis of algorithms used in machine learning. Currently,  heuristic approaches are popular because the solution concepts are believed to be intractable. Our results, for example our algorithm for NMF under the separability condition, raise hope that sometimes the solution concepts may not be intractable after all.

\section*{Acknowledgements}

We thank David Blei and Saugata Basu for useful discussions.

\newpage
\appendix

\section{Extended Discussion}\label{sec:appendix:remarks}

Here we explain an application of NMF in detail. Perhaps the best approach is to contrast the NMF problem with a more well-known matrix factorization, the singular value decomposition (SVD): A $n \times m$ matrix $M$ can be written as $M = \sum_i \sigma_i u_i v_i^T$ where the set $\{u_i\}_i$ and the set $\{v_i\}_i$ are orthonormal and $\sigma_1 \geq \sigma_2 .... \geq \sigma_r > 0$ (see e.g. \cite{GV}). In a number of applications, we imagine that the columns of $M$ represent examples and the rows of $M$ represent observed variables. In the context of information retrieval one often forms $M$ as a "term-by-document" matrix where the $(i, j)^{th}$ entry in $M$ is the frequency of occurrence of the $i^{th}$ term in the $j^{th}$ document in the database. The SVD of $M$ (e.g. in Latent Semantic Indexing (LSI) \cite{LSI}) is often interpreted as a method to extract "topics" in the database: The set of vectors $\{u_i\}_i$ (in a truncated SVD) is the subspace that contains the maximum variance of the documents, and projecting columns of $M$ (i.e. documents) onto this basis is interpreted as a decomposition of each document into constituent topics. Documents can then be compared based on an inner-product in this space. 

In some sense, the decomposition into "topics" generated via SVD is inconsistent with our intuitive notion of what a topic is. The vectors $\{u_i\}_i$ have both positive and negative values -- these vectors are orthogonal. For example, imagine some documents are about cars and some others are about the weather. These "topics" would both be negatively correlated with mentioning the word "elephant" -- i.e. documents about either topic are unlikely to use this word. What this means is that when we compute the similarity of a pair of documents, the documents will be judged to be more similar if both omit the word "elephant". But this is not consistent with our intuitive model, and would lead to spurious latent relationships. We would expect similarity to be based on positive occurrences only. 

Hofman introduced a related approach (Probabilistic Latent Semantic Indexing \cite{Hof}) in which each document is normalized to be a distribution on words, and the goal is to compute a small set of $r$ topics (which are each distributions on words) and represent each document as a distribution on topics. This is equivalent (after an appropriate renormalization) to computing a nonnegative factorization of the term-by-document matrix $M$ into $AW$, where the columns of $A$ represent a set of $r$ topics and each column of $W$ expresses the corresponding document as a distribution on topics. The advantage of requiring this factorization to be nonnegative is that in Hofman's PLSI documents are judged to be similar based on words that they both contain. In LSI, documents can also be judged to be similar based on words they both omit. Arguably, Hofman's model is more consistent with our intuition and maybe this helps explain why (computational issues aside) a nonnegative factorization is, in many cases, preferred over an unrestricted one. 

\end{document}